\documentclass[10pt]{amsart}

\usepackage{graphicx} % Required for inserting images
\usepackage{color}
\usepackage{pdfpages}
\usepackage{enumerate}
\usepackage{hyperref}
\usepackage{float}
\usepackage{verbatim}
\usepackage{tabularx}
\usepackage{bbm}  % for \mathbbm{k}

\newcommand{\mr}{\mathrm}
\newcommand{\ra}{\rightarrow}
\newcommand{\tr}{\mathrm{tr}}
\newcommand{\CC}{\mathbb{C}}
\newcommand{\RR}{\mathbb{R}}
\newcommand{\HH}{\mathcal{H}}
\newcommand{\xra}{\xrightarrow}
\newcommand{\til}{\widetilde}
\newcommand{\bt}{\bullet}
\newcommand{\bl}{\color{blue}}
\newcommand{\Geom}{\mathrm{Geom}}
\newcommand{\sGeom}{\mathrm{sGeom}}
\newcommand{\mc}{\mathcal}
\newcommand{\MM}{\mathcal{M}}
\newcommand{\tot}{\mathrm{tot}}
\newcommand{\ol}{\overline}
\newcommand{\fr}{\mathrm{fr}}
\newcommand{\Hom}{\mathrm{Hom}}
\newcommand{\ii}{\mathrm{in}}
\newcommand{\oo}{\mathrm{out}}
\newcommand{\gr}{\color{green}}
\newcommand{\flip}{\mathrm{flip}}

\newcommand{\wh}{\widehat}
\newcommand{\cc}{\mathsf{c}}
\newcommand{\dd}{\partial}
\newcommand{\ccirc}{\circ_c}
\newcommand{\ZZ}{\mathbb{Z}}
\newcommand{\sym}{\mr{sym}}
\newcommand{\Sp}{\mathrm{sp}}
\newcommand{\Conv}{\mathrm{Conv}}

\newcommand{\p}{\mathsf{p}}
\newcommand{\Adm}{\mathrm{Adm}}
\newcommand{\Ahat}{\widehat{A}_\infty}
\newcommand{\SP}{\mathsf{SP}}
\newcommand{\kk}{\mathbbm{k}}
\newcommand{\D}{\mathcal{D}}
\newcommand{\cyl}{\mathrm{cyl}}
\newcommand{\BV}{\mathrm{BV}}
\newcommand{\rd}{\mathrm{red}}

\theoremstyle{remark}
\newtheorem{remark}{Remark}[section]
\theoremstyle{plain}

\newtheorem{proposition}[remark]{Proposition}
\newtheorem{thm}[remark]{Theorem}
\newtheorem{corollary}[remark]{Corollary} 
\theoremstyle{definition}
\newtheorem{definition}[remark]{Definition}
\newtheorem{example}[remark]{Example}

\newtheorem{conjecture}[remark]{Conjecture}

\newtheorem{question}[remark]{Question}

\hypersetup{
           breaklinks=true,   % splits links across lines
        }

\title[Combinatorial 2d HTQFT from a local cyclic $A_\infty$ algebra]{
Combinatorial 2d higher topological quantum field theory from a local cyclic $A_\infty$ algebra
}
\author[J.Beck]{Justin Beck}
\address{University of Notre Dame, Notre Dame, IN 46556, USA}
\email{jbeck6@nd.edu}

\author[A.Losev]{Andrey Losev}

\address{Shanghai Institute for Mathematics and Interdisciplinary Sciences, Building 3, 62 Weicheng Road, Yangpu District, Shanghai, 200433, China}
\email{%aslosev2@gmail.com
aslosev2@yandex.ru
}

\author[P.Mnev]{Pavel Mnev}

\address{University of Notre Dame, Notre Dame, IN 46556, USA}
\email{pmnev@nd.edu}

\begin{document}

\begin{abstract}
    We construct combinatorial analogs of 2d higher topological quantum field theories. 
  We consider triangulations as vertices of a certain CW complex $\Xi$. In the ``flip theory,'' cells of $\Xi_\flip$ correspond to polygonal decompositions obtained by erasing the edges in a triangulation.
  These theories assign to a cobordism $\Sigma$ a cochain $Z$ on $\Xi_\flip$ constructed as a contraction of structure tensors of a cyclic $A_\infty$ algebra $V$ assigned to polygons. The cyclic $A_\infty$ equations imply the closedness equation $(\delta+Q)Z=0$. In this context we define combinatorial BV operators and give examples with coefficients in $\ZZ_2$.

  In the ``secondary polytope theory,'' $\Xi_\mr{sp}$ is the secondary polytope (due to Gelfand-Kapranov-Zelevinsky) and the cyclic $A_\infty$ algebra has to be replaced by an appropriate refinement that we call an $\Ahat$ algebra.

  We conjecture the existence of a good Pachner CW complex $\Xi$ for any cobordism, whose local combinatorics is descibed by secondary polytopes and the homotopy type is that of Zwiebach's moduli space of complex structures. Depending on this conjecture, one has an ``ideal model'' of combinatorial 2d
HTQFT determined by a local $\widehat{A}_\infty$ algebra.
\end{abstract}

\maketitle

\setcounter{tocdepth}{3}
\tableofcontents

\section{Introduction}
Simplest TQFTs are defined as quantum field theories that do not depend on geometrical data on the manifold. Refined versions contain a homological symmetry (i.e. the target of the QFT functor is the category of chain complexes rather than vector spaces) such that this dependence is a exact. We propose to study a new class of models where the manifold is equipped with a triangulation %cellular decomposition 
playing the role of geometry. Changes of this triangulation (generated by Pachner moves) are required to be exact with respect to the differential, see Figure \ref{fig: intro Pachner}. 
%\marginpar{new}
\begin{figure}[h]
    \centering
    \includegraphics[scale=0.5]{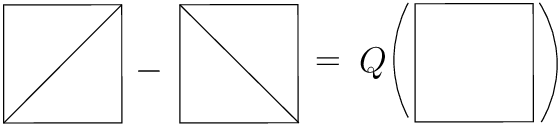}\hspace{2cm} \includegraphics[scale=0.5]{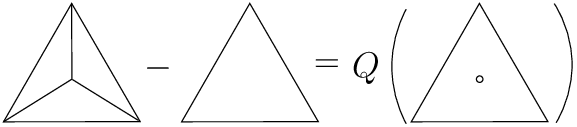}
    \caption{Exactness with respect to Pachner's moves, schematically.}
    \label{fig: intro Pachner}
\end{figure}

In this paper we study the 2d case (and 1d case as a toy model). Here partition functions are assigned to triangulated surfaces. The space of triangulations $\Xi$ itself may be considered as a CW complex (we also call it the Pachner complex), where:
\begin{itemize}
\item 
Vertices of $\Xi$ are %particular 
triangulations of the surface. 
\item  Edges of $\Xi$ correspond to Pachner moves (Figure \ref{fig: Pachner moves}). An edge corresponding to a Pachner flip can be seen as an insertion of one square (where the flip occurs) in a triangulation, cf. Figure \ref{fig: 1-cell of Xi_flip}. An edge corresponding to the second Pachner move -- a stellar subdivision/aggregation -- can be seen as having one triangle with a ``floating point''  in the triangulation (see the discussion of secondary polytopes in the paper and Figure \ref{fig: secondary polytope - triangle with one floating point}). 
\item There are also higher cells of $\Xi$ corresponding to relations between Pachner's moves (see Figures \ref{fig: 2-cells in Xi_flip}, \ref{fig: SP triangle with 2 floating points}, \ref{fig: SP square with 1 floating point}). 
%\marginpar{Put a picture here?}
\end{itemize}

In this paper we present two models of such a theory, and conjecture the existence of the ideal model combining the features of the two.
%\begin{enumerate}[1.]

    \textbf{Model 1 (``flip theory''):} Only Pachner flips are considered, in particular the set of vertices is fixed. We denote the corresponding CW complex $\Xi_\flip$. The cells of $\Xi_\flip$ correspond to polygonal decomposition of the surface, where the dimension of the cell $e$ is 
    \begin{equation}
    \dim(e)=\sum_i (n_i-3),
    \end{equation} 
    where the sum is over the polygons and $n_i$ is the number of sides of $i$-the polygon. The cell $e$ can be identified with a product of Stasheff's associahedra $\prod_i K_{n_i-1}$.
    
\textbf{Model 2 (``secondary polytope theory''):} For $A\subset \RR^2$ a collection of points in general position, one has a convex polytope $\Xi_\mr{sp}$ -- the ``secondary polytope'' introduced in \cite{GKZ}. Vertices of $\Xi_\mr{sp}$ correspond to \emph{regular}\footnote{
A triangulation is regular if one can find a continuous convex function linear on the triangles and breaking (non-differentiable) on the edges.
} triangulations of the polygon $\Sigma$ -- the convex hull of $A$, with vertices in a subset of $A$. Edges of $\Xi_\mr{sp}$ correspond to Pachner moves between regular triangulations and there are certain higher cells. This is a partial model as it is defined only on a disk, it relies on convex geometry of $\RR^2$ and not only on topology, and the maximal set of vertices $A$ is fixed (one can remove vertices via stellar aggregations but one cannot add vertices not in $A$).

\textbf{The ideal Model 3:} We conjecture the existence of an ideal model combining the features of Models 1 and 2. Namely, it should be defined for all surfaces and allow both Pachner moves. We also conjecture that the homotopy type of $\Xi$ in the ideal model is that of the moduli space of complex structures decorated with jets of local coordinates at punctures (Zwiebach's moduli space $\til\MM_{h,n}$).
In this correspondence, boundary circles of the surface (``polygonal holes'') correspond to the punctures.
%\end{enumerate}

After we described the space of triangulations, we proceed to the construction of functors from cobordism category to the category of chain complexes. The value of the functor on a cobordism is a
cochain on the ``space of triangulations'' $\Xi$ taking values in
 linear maps between spaces of in- and out-states.  
 In particular, in Model 1, we fix a cyclic $A_\infty$ algebra $V$ with differential $Q=m_1$ and other operations $m_2,m_3,\ldots$ and bilinear pairing $g$. Cyclicity of  an $A_\infty$ algebra with respect to the pairing $g$ means that operations 
\begin{equation}c_{n+1}=g(\bt,m_n(\bt,\ldots,\bt)) \colon V^{n+1}\ra \kk\end{equation} 
are invariant with respect to cyclic permutations of inputs.
The space of states for a circle triangulated into $k$ intervals is just $V^{\otimes k}$ considered as a cochain complex. The partition function for a surface equipped with a polygonal decomposition $\alpha$ is
\begin{equation}\label{intro Z(alpha)}
    Z(\alpha)=\left\langle \bigotimes_{\mr{polygons}\, p\, \mr{of}\, \alpha} c_{|p|}, \bigotimes_{\mr{edges\,of\,}\alpha} g^{-1}  \right\rangle
\end{equation}
-- the contraction of elements
$c_{|p|}\in (V^*)^{\otimes |p|}$ assigned to polygons ($|p|$ is the number of sides of the polygon $p$) -- the cyclic $A_\infty$ operations -- and inverse pairings attached to the edges. 
Since polygonal decompositions correspond to cells of $\Xi_\flip$, the formula (\ref{intro Z(alpha)}) determines the value of a cochain on $\Xi_\flip$ on these cells.

By virtue of higher associativity relations in the $A_\infty$ algebra $V$, the cochain $Z$ satisfies the universal closedness equation
\begin{equation}\label{intro main eq}
    (\delta+Q)Z=0,
\end{equation}
with $\delta$ the coboundary operator on cochains!

The equation (\ref{intro main eq}) in degree zero  means that $Z$ for a triangulated surface is $Q$-closed. In degree one it means that for $\alpha_1,\alpha_2$ two triangulations that differ by a Pachner flip, and for $\alpha$ the same triangulation with a square inserted where the flip is occurring (see Figure \ref{fig: 1-cell of Xi_flip}), one has 
\begin{equation}\label{intro Z(flip)}
    Z(\alpha_1)-Z(\alpha_2)=Q Z(\alpha).
\end{equation}
%-- the difference of partition functions for the traingulations is $Q$ of the partition function for $\alpha$.

If we think of the flip as an analog of a change of complex structure on the surface then the l.h.s. of (\ref{intro Z(flip)}) corresponds to insertion of the stress-energy tensor $T_{zz}$ while the $Z(\alpha)$ in the r.h.s. corresponds to insertion of its superpartner $G_{zz}$. The meaning of (\ref{intro Z(flip)}) is just 
\begin{equation}
T_{zz}=Q(G_{zz}).
\end{equation}

In Section \ref{sec: improved BV operator}, in the context of Model 1 we define a combinatorial version of the BV operator $G_{0,-}$ as the value of $Z$ on a ``BV-cycle'' -- a 1-cycle of $\Xi_\flip$ on a cylinder that squares to a boundary. We construct examples of such BV cycles with coefficients in $\ZZ_2$.

In Model 2, the spaces of states are the same and the partition function is defined by the same formula (\ref{intro Z(alpha)}). However, $\alpha$ (the decoration of a cell in $\Xi_\mr{sp}$) now is a polygonal decomposition with floating points (with a prescribed ``configuration chamber'' describing their geometry rather than topology). Thus, instead of cyclic $A_\infty$ operations on $V$ we need a more refined ``$\Ahat$'' algebra structure (see Section \ref{ss: def Ahat algebra}) where multilinear operations correspond to polygons with floating points in a configuration chamber. Together with new operations (which include the usual $A_\infty$ operations), we have new relations, coming from the adjacency of cells in secondary polytopes. These relations imply that the partition function seen as a  cochain on the secondary polytope again satisfies the universal closedness equation (\ref{intro main eq}).

In the toy model of dimension one, Model 2 corresponds to quantum mechanics with $Q$-exact Hamiltonian, see Section \ref{ss: 1d HTQM from secondary polytopes}. On the other hand, in dimension two the stellar subdivision equation is an analog of $Q$-exactness of the trace of the stress-energy tensor.

%\subsection{Analogy between continuum and combinatorial 2d field theory}  \marginpar{Put someplace later?}
%{\gr Flip looks like  a change of complex structure. Generator for the square superpartner of off-diagonal components of stress-energy tensor.
%Stellar subdivision looks like Weyl transformation. Point field insertions (point observables) $\sim$ states inserted in a polygonal hole. $c_3\sim$ vacuum vector. $c_4\sim G_{zz}(dz)^2+G_{\bar{z}\bar{z}}d\bar{z}^2$, $c$ of triangle with a dot $\sim$ $G_{z\bar{z}}dz d\bar{z}$}
%\leavevmode \\
In the table below we summarize the analogy between the continuum and combinatorial 2d field theory.

%\begin{tabularx}{330pt}{c|c} 
%% using tabularx instead of tabular to allow footnote
\begin{tabular}{c|c}
    \textbf{Combinatorial HTQFT} & \textbf{Continuum TCFT} \\ \hline \hline
    Pachner flip &  deformation of complex structure \\ \hline
    stellar subdivision & Weyl transformation of metric \\ \hline \hline
    Pachner complex $\Xi$ & Zwiebach's moduli space $\til\MM_{h,n}$
 %   \footnote{
%The canonical $(S^1)^n$-bundle over the (non-compactified) moduli space of complex structures on a surface of genus $h$ with $n$ ordered punctures.
%    } 
    \\ \hline
    Pachner complex in genus zero & framed little disk operad $E_2^\fr$\\
    \hline \hline
    state on a polygon &   quantum field $\slash$ point observable \\ \hline
    partition function paired with  states  & correlator of observables\\
    on the boundary of  polygonal holes & \\  
     \hline 
    partition function &  partition function \\
    is a $(\delta+Q)$-closed & is a $(d+Q)$-closed \\
    cochain on $\Xi$ &  form on the moduli space \\
    \hline \hline
    cyclic product $c_3$ & vacuum vector  \\
    & (partition function of a disk) \\ \hline
    cyclic $A_\infty$ %Massey 
    operation $c_4$ &  field $G_{zz}(dz)^2+G_{\bar{z}\bar{z}}(d\bar{z})^2$ \\ \hline
    %& (traceless part of the superpartner of stress-energy) \\
    $\wh{A}_\infty$ operation $\mu$  
   & field $G_{z\bar{z}}dz d\bar{z}$\\
    (triangle with a point inside)  & \\ \hline
   % & (trace of the stress-energy tensor) \\
    defect of associativity $Q(c_4)$ & stress-energy tensor (traceless part) \\ \hline
    defect of stellar subdivision $Q(\mu)$ & trace of stress-energy tensor
\end{tabular}
%\end{tabularx}

%{\gr Also: open string algebra/ closed string algebra?}

\subsection{Main results}
\begin{enumerate}[(i)]
    \item Given a cyclic $A_\infty$ algebra, formula (\ref{intro Z(alpha)}) (the partition function of the ``flip model'' of combinatorial HTQFT) defines a $(\delta+Q)$-closed cochain on the flip complex and is functorial with respect to gluing of cobordisms -- Theorem \ref{thm: model 1}.
    \item Within the the flip model, we define a natural notion of a BV cycle -- a 1-cycle on the flip complex of a cylinder whose composition-square is zero in homology -- and give nontrivial examples, see Proposition \ref{prop: c^k}.
    
    As an aside, in Section \ref{ss: BV_infty} we give an example where a BV cycle can be extended to a nonhomogeneous cycle which satisfies Maurer-Cartan equation at the chain level (Proposition \ref{prop: BV_infty}), giving rise to a combinatorial replacement of the $S^1$-equivariant BRST operator  operator $Q+u G_{0,-}$ in TCFT.
    \item We construct the ``secondary polytope theory'' as a $(\delta+Q)$-closed cochain on the secondary polytope of a finite set  $A\subset \RR^2$ (understood as a model of combinatorial HTQFT on a disk realized as the convex hull of $A$, with $A$ the allowed vertices of triangulations), see Theorem \ref{thm: model 2}. 
    
    The local input of the construction is an ``$\Ahat$ algebra'' (defined in Section \ref{ss: def Ahat algebra}) -- a refinement of the notion of a cyclic $A_\infty$ algebra by certain extra homotopies, encoding in particular compatibility with the second 2d Pachner move (stellar subdivision/aggregation).
    \item As an aside, we develop a 1d toy model of the secondary polytope theory where the secondary polytopes are simply cubes -- Section \ref{ss: 1d HTQM from secondary polytopes}. In particular, in the continuum limit, the corresponding HTQFT becomes an interesting generalization of the standard topological quantum mechanics (a.k.a. twisted $\mathcal{N}=2$ supersymmetric quantum mechanics) -- Example \ref{ex: infinitesimal SP_1 algebra}.
    \item Subject to a conjecture on existence of a ``good'' Pachner complex (Conjecture \ref{conjecture: Pachner complex}), 
    \begin{enumerate}[(a)]
        \item We construct a $(\delta+Q)$-closed cochain on the Pachner complex, behaving functorially with respect to gluing of cobordisms -- Corollary \ref{corollary: Model 3}.
    \item We construct a canonical BV algebra structure on the center of the cohomology of an $\Ahat$ algebra (the reduced space of states on a circle) -- Corollary \ref{cor: BV algebra on the reduced space of states}.
      \end{enumerate}
\end{enumerate}

\subsection*{Acknowledgements} 
%We thank Kirill Salmagambetov for discussions at an early stage of the project (in 2017) and Lev Soukhanov for an illuminating talk on secondary polytopes. P.M. thanks Nikolai Mnev for inspiring discussions.
We thank Nikolai Mnev and Lev Soukhanov for inspiring discussions. We also thank Kirill Salmagambetov for discussions at an early stage of the project (in 2017).

%\section{Combinatorial 2d HTQFT: the idea}
\section{Warm up: a combinatorial 2d TQFT (strict case)}
\label{sec: combinatorial strict 2d TQFT}

%\marginpar{Put in the intro/preamble/Notations/disclaimer?}
Throughout the paper $\kk$ will stand for the ground field. By default, $\kk=\CC$ or $\RR$. %(we can also consider $\kk=\ZZ_2$ when we want to ignore signs; 
%where the relevant formulae have no even denominators).
In Section \ref{sec: improved BV operator} we will have results on combinatorial BV operators which require $\kk=\ZZ_2$.
\subsection{Local data} 
%Fix a ground field $\kk$; 
Fix a finite-dimensional associative algebra $V$ over $\kk$ with multiplication $m_2\colon V\otimes V\ra V$ and with nondegenerate invariant pairing ${g\colon V\otimes V\ra \kk}$. Denote $c_3\colon V^{\otimes 3}\ra \kk$ the corresponding 3-to-1 operation $c_3(x,y,z)=g(x,m_2(y,z))$; it is invariant under cyclic permutations of inputs.
%$g(x,y)=\tr_V m_2(x,m_2(y,\bullet))$.
\subsection{Partition function} Let $\Sigma$ be an oriented compact surface, 
%(a 2d manifold\footnote{We will be working in the PL (piecewise-linear) category.}), 
possibly with boundary, equipped with a triangulation $T$ (Figure \ref{fig: triangulated surface}). Assume that $\Sigma$ has $n$ boundary circles $S^1_1,\ldots,S^1_n$ subdivided by $T$ into $k_1,\ldots,k_n$ intervals.  
\begin{figure}[h]
\includegraphics[scale=0.7]{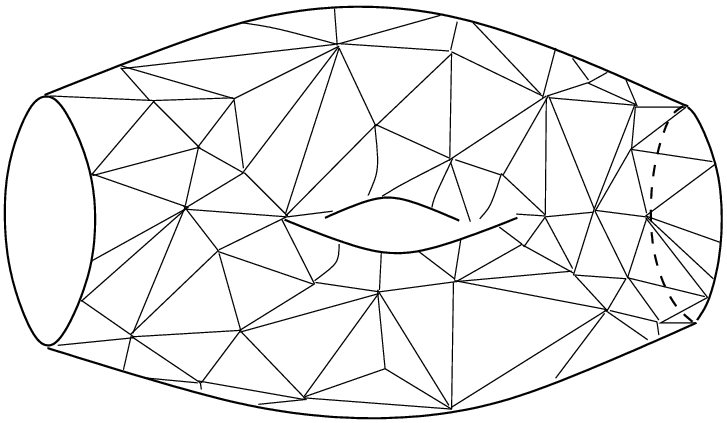}
\caption{Triangulated surface.}
\label{fig: triangulated surface}
\end{figure}
We define the spaces of states for the boundary circles as 
\begin{equation}
\HH(S^1_i,T|_{S^1_i})=V^{\otimes k_i} 
\end{equation} 
and the partition function for the surface as the contraction of tensors $c_3$ assigned to the triangles (2-simplices of $T$), using tensors $g^{-1}\in V\otimes V$ assigned to the %non-boundary 
edges (1-simplices of $T$): 
\begin{equation}\label{Z strict}
Z(\Sigma,T)=\left\langle \bigotimes_{\mr{triangles}}c_3, \bigotimes_{\mr{%internal\; 
edges}}g^{-1} \right\rangle \qquad \in \bigotimes_{i=1}^n
V^{\otimes k_i} = \bigotimes_{i=1}^n \HH(S^1_i,T|_{S^1_i}).
\end{equation}
Here we note that the edges of $T$ which lie on the boundary of $\Sigma$ are not contracted, so the result is a vector in the tensor product of spaces assigned to all boundary edges. If $\Sigma$ is a closed surface, $Z(\Sigma,T)$ is a number. 

If some boundary circles are assigned as ``in-boundary'' and some as ``out-boundary,'' one thinks of $\Sigma$ as a cobordism. Then one can reinterpret (using the inner product on spaces of states induced by $g$) the partition function (\ref{Z strict}) as a linear operator 
-- an element of $\mr{Hom}(\HH_\ii,\HH_\oo)$.

\subsection{Functoriality (cutting-gluing property)}\label{ss: functoriality of strict tqft}
The construction above gives rise to a QFT as a monoidal functor from the category of triangulated 2-cobordisms  (the objects in this category are disjoint unions of triangulated circles) 
to the category $\mr{Vect}$ of vector spaces. In particular, if one has two composable triangulated cobordisms
$ ((S^1)^{\sqcup n_1},\tau_1) \xra{(\Sigma',T')} ((S^1)^{\sqcup  n_2},\tau_2)$  and 
$ ((S^1)^{\sqcup n_2},\tau_2) \xra{(\Sigma'',T'')} ((S^1)^{\sqcup n_3},\tau_3)$, then their for gluing $ ((S^1)^{\sqcup n_1},\tau_1) \xra{(\Sigma,T)} ((S^1)^{\sqcup n_3},\tau_3)$ one has
\begin{equation}
Z(\Sigma,T)=  Z(\Sigma'',T'') \circ Z(\Sigma',T') \qquad \in \Hom(\HH_1, \HH_3).
%\left\langle Z(\Sigma',T')\otimes Z(\Sigma'',T''), \bigotimes_{i=1}^{n_2} g^{-1} \right\rangle. 
\end{equation}

\subsection{Invariance with respect to triangulations}
By the celebrated Pachner's theorem \cite{Pachner},\footnote{
In fact, here we are using its relative version, see \cite{Casali}.
} one can transition between two triangulations $T,T'$ of the same surface with boundary (assume that $T,T'$ induce the same triangulation of the boundary) by a sequence of ``bistellar moves,'' a.k.a. ``Pachner moves'' (Figure \ref{fig: Pachner moves}): 
\begin{enumerate}[(1)]
    \item flips (switching the diagonal in a square);
    \item stellar subdivisions/aggregations -- subdividing a triangle into three triangles, or the inverse (merging) move.
\end{enumerate}
\begin{figure}[h]
    \centering
    \includegraphics[scale=0.7]{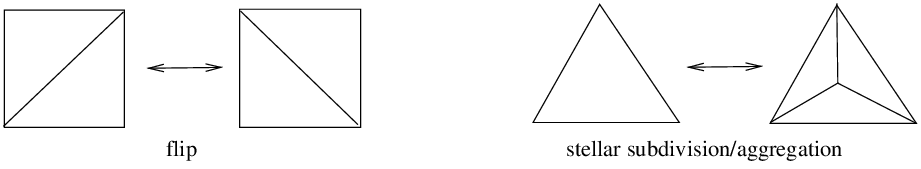}
    \caption{2d Pachner moves.}
    \label{fig: Pachner moves}
\end{figure}

Invariance of partition the function $Z(\Sigma,T)$ w.r.t. flip on $T$ follows from associativity of $m_2$:
\begin{multline}\label{flip invariance strict tqft}
\underbrace{\left\langle g(\bt_1,m_2(\bt_2,m_2(\bt_3,\bt_4)))\otimes Z(\til\Sigma,\til{T}), (g^{-1})^{\otimes 4} \right\rangle}_{Z(\Sigma,T)}= \\
=\underbrace{\left\langle g(\bt_1,m_2(m_2(\bt_2,\bt_3),\bt_4)))\otimes
Z(\til\Sigma,\til{T}),(g^{-1})^{\otimes 4} \right\rangle}_{Z(\Sigma,T_\mr{flipped})}, 
\end{multline}
see Figure \ref{fig: associativity}.
Here $\til\Sigma,\til{T}$ is the complement of the square where the flip is performed, and its triangulation.
%[DRAW ASSOCIATIVITY VIA DUAL RIBBON GRAPHS]
\begin{figure}[h]
    \centering
    \includegraphics{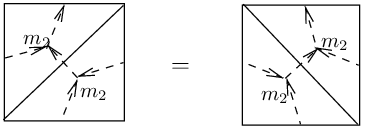}
    \caption{Associativity as a move on ribbon graphs dual to triangulations.}
    \label{fig: associativity}
\end{figure}

To have invariance w.r.t. stellar subdivisions/aggregations, one needs to have
\begin{equation}\label{Pachner 2 rel (strict)}
\tr_V m_2(x,m_2(y,m_3(z,\bt)))=c_3(x,y,z),\quad x,y,z\in V,
\end{equation}
see Figure \ref{fig: Pachner 2 relation}.
\begin{figure}[h]
    \centering
    \includegraphics[scale=0.7]{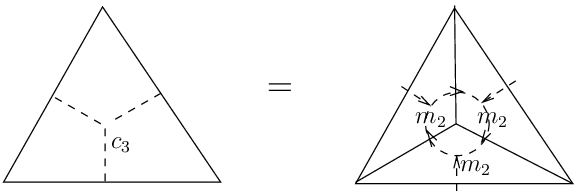}
    \caption{Relation (\ref{Pachner 2 rel (strict)}).}
    \label{fig: Pachner 2 relation}
\end{figure}
This property is not automatic. 
%It holds if we require the metric to be inferred from the multiplication via
However it is true if the metric is related to $m_2$ by
\begin{equation}\label{g from m2}
\tr_V m_2(x,m_2(y,\bt))=g(x,y), 
\end{equation}
see Figure \ref{fig: g}.
%\marginpar{new}
\begin{figure}[h]
    \centering
    \includegraphics[scale=0.8]{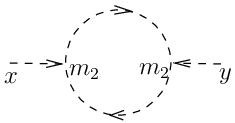}
    \caption{Metric $g$ in terms of $m_2$.}
    \label{fig: g}
\end{figure}

To summarize,  partition functions $Z(\Sigma,T)$ are invariant w.r.t. Pachner flips. Assuming the extra property (\ref{g from m2}) for the algebra $V$, we have that partition functions $Z(\Sigma,T)$ are independent of the triangulation in the bulk.  In this case the model defines a 2d TQFT as a monoidal functor from the category of non-triangulated cobordisms to the category of vector spaces. 

\subsubsection{Space of states for a non-triangulated circle}\label{sss: strict TQFT: space of states for non-triangulated circle}
The space of states for a non-triangulated circle can be constructed as
the space of equivalence classes of tensors in $V^{\otimes k}$ for $k=1,2,\ldots$ where the equivalence relation is generated by partition functions for triangulated cylinders.\footnote{
 Put another way, it is the colimit of the diagram in $\mr{Vect}$ given by all triangulated cylinders.
} Equivalently,
\begin{equation}\label{H=HH}
\HH(S^1)=HH_0(V,V)=V/[V,V]\simeq\mc{Z}
\end{equation}
is the zeroth Hochschild homology of $V$ with coefficients in $V$, or, equivalently, the center of $V$.
%\marginpar{maybe reference to Lurie?}
%\textcolor{red}{[Add picture: triangulated cylinders]}

  %\marginpar{This came out a bit long. Move to an appendix?}
\textbf{Explanation of (\ref{H=HH}):} First, note that the algebra $V$ admits an orthogonal decomposition  
(w.r.t. the metric (\ref{g from m2})) into the commutant and the center,\footnote{
Indeed, $\lambda \in [V,V]^\perp \Leftrightarrow g(\lambda,[x,y])=0\; \forall x,y\in V \Leftrightarrow g([\lambda,x],y)=0\; \forall x,y\in V \Leftrightarrow [\lambda,x]=0\; \forall x\in V \Leftrightarrow \lambda\in \mc{Z}$. Thus, $[V,V]^\perp=\mc{Z}$. Note that this argument only uses that $g$ is an invariant nondegenerate metric on $V$, it does not use the particular form (\ref{g from m2}).
} 
\begin{equation}\label{V=[V,V]+Z}
V=[V,V]\oplus\mc{Z}.    
\end{equation} 
Consider %the equivalence relation given by 
the simplest triangulated cylinder (Figure \ref{fig: simple cylinder}). 
\begin{figure}[h]
    \centering
    \includegraphics{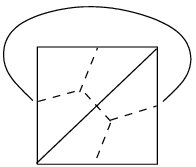}
    \caption{Triangulated cylinder (left and right sides of the square are glued).}
    \label{fig: simple cylinder}
\end{figure}
Treating both boundaries as ``in,'' its value is the pairing
\begin{equation}\label{mu cylinder}
    \begin{array}{cccc}
         \gamma\colon& V\otimes V&\ra &\kk  \\
          & (x,y)&\mapsto & \tr_V m_2(x,m_2(\bt,y))  
    \end{array}
\end{equation}
%The algebra $V$ admits an orthogonal decomposition (w.r.t. the metric (\ref{g from m2})) into the commutant and the center, 
%\begin{equation}\label{V=[V,V]+Z}
%V=[V,V]\oplus\mc{Z}.    
%\end{equation} 
The pairing (\ref{mu cylinder}) is symmetric, has $[V,V]$ as its kernel and coincides with $g$ when restricted to $\mc{Z}$. Hence, the 
partition function of this cylinder, with one boundary circle seen as ``in'' and the other as ``out,'' is the projection to $\mc{Z}$ in (\ref{V=[V,V]+Z}).
The equivalence relation induced by this  cylinder %(with one boundary circle seen as ``in'' and the other as ``out'') 
says that two elements $x,y\in V$ are equivalent if their projections to $\mc{Z}=V/[V,V]$ coincide. More generally, the equivalence relation is: 
\begin{equation}\label{equiv rel}
\otimes_{i=1}^k x_i \sim \otimes_{j=1}^l y_j \quad \mbox{if}\quad P(x_1\cdots x_k)=P(y_1\cdots y_l)
\end{equation} 
where $P$ is the projection onto $\mc{Z}$ in the decomposition (\ref{V=[V,V]+Z}) and the products inside are the iterated $m_2$-products. E.g. by considering the triangulated cylinder
$$ \vcenter{\hbox{\includegraphics{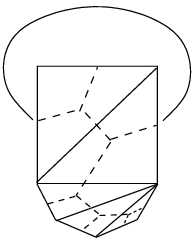}}} $$
we see the equivalence relation (\ref{equiv rel}) for $k$ any and $l=1$; by composing such a cylinders with an upside-down one, we obtain the full relation (\ref{equiv rel}). By triangulation invariance of partition functions relative to the boundary, the triangulation of the bulk of the cylinder does not matter. %\marginpar{Edit/check the wording}

\subsection{A famous example: 2d Dijkgraaf-Witten model}
Fix $G$ a finite group and $V=\CC[G]$ its group algebra (in this subsection we set $\kk=\CC$). %ring. 
The corresponding TQFT is known as 2d Dijkgraaf-Witten model \cite{DW}. Its partition functions give the volume of the groupoid of $G$-coverings of $\Sigma$.
%\footnote{
%This model is also meaningful for $G$ a Lie group, then the model is the 2d Yang-Mills at zero coupling constant (a.k.a. $BF$ theory). 
%} 
In particular, the partition function for a closed orientable surface $\Sigma$ of genus $h$ is
\begin{equation}\label{DW sum over irreps}
Z(\Sigma)=\frac{|\mr{Hom}(\pi_1(\Sigma),G)|}{|G|} = |G|^{2h-2}\sum_{R\in \mr{irrep}(G)}
\dim(R)^{2-2h}.
\end{equation}

\begin{remark}\label{rem: 2dBF} 
A finite group $G$ can be replaced by a compact Lie group, with $V=L^2(G)$  -- the space of square-integrable functions of $G$ equipped with convolution product. This is the 2d Yang-Mills theory at zero coupling constant, a.k.a. 2d $BF$ theory, cf. \cite{Migdal,Witten 2dYM}. Its partition function for a closed surface is the symplectic volume of the moduli space of flat $G$-bundles on $\Sigma$ and is again given by the sum in the r.h.s. of (\ref{DW sum over irreps}) without $|G|^{2h-2}$ prefactor. Note the space of irreducible representations in this case is infinite and the sum over irreducible representations in (\ref{DW sum over irreps}) is divergent for genus $h=0,1$ (but converges for $h\geq 2$).
\end{remark}

\begin{remark}[Area deformation]
 %Set $\kk=\CC$.
One can deform the construction of this section by equipping triangles with areas
%nonnegative areas 
(seen as a geometric structure) and by allowing the product $m_2$ on $V$ to depend on the area $A\geq 0$ of a triangle, 
%with iterated products depending on the total area of a triangulated surface, 
with associativity relation $m_2^{A_1}(m_2^{A_2}(x,y),z)=m_2^{A_3}(x,m_2^{A_4}(y,z))$ for any $x,y,z\in V$ and areas of triangles satisfying $A_1+A_2=A_3+A_4$.
In this setting the partition function (\ref{Z strict}) of a triangulated surface depends on the total area -- the sum of areas of the triangles.\footnote{See \cite{Runkel}.}  %\marginpar{Cite Runkel?}

For $V=\CC[G]$ such a deformation corresponds to weighing the summand in the r.h.s. of (\ref{DW sum over irreps}) with the factor $e^{- C(R)\mr{Area}(\Sigma)}$, with $\mr{Area}(\Sigma)$ the total area of the surface, %(sum of areas of triangles), 
and $C(R)$ some constants assigned to irreducible representations $R$ (these constants parametrize the possible area deformations).\footnote{ In this example, the deformed product in the group algebra is defined on generators by $m_2^A(g_1,g_2)=\sum_{g_3\in G}\sum_{R\in\mr{irrep}(G)} e^{-A C(R)}\dim(R)\chi_R(g_1 g_2 g_3^{-1}) g_3$, where $g_1,g_2$ are elements in $G$, $\chi_R$ is the character of the representation $R$ and $A\geq 0$ is the area of the triangle. 
%(parameter of the deformation).
}

In the setting of Remark \ref{rem: 2dBF}, such a deformation with $C(R)$ being the quadratic Casimir of the Lie algebra $\mr{Lie}(G)$ in the representation $R$ yields the 2d Yang-Mills theory (at nonzero coupling constant).
%   {\gr  [AREA DEFORMATION] }
\end{remark}

\section{HTQFTs vs. strict TQFTs}
\subsection{HTQFT setup}
In a functorial QFT we equip a cobordism $\Sigma$ with extra  geometric data such as metric, complex structure, triangulation, etc. 

E.g. in dimension 1 the cobordism is an interval and we can equip it with length. The image of the functor  applied to 
an interval (the partition function of the interval) is $Z(t)\in
%\mr{End}(\HH) \otimes \mr{Fun}(\RR_+)
\mr{Map}(\RR_+,\mr{End}(\HH))
$,  where $\HH$ is the space of states for a point. 
Functoriality implies 
$${Z(t_1+t_2)=Z(t_1)\circ Z(t_2)}.$$
That implies 
\begin{equation}
Z(t)=\exp(tH)
\end{equation} 
for some linear operator $H\in \mr{End}(\HH)$ -- the Hamiltonian.\footnote{
In high energy physics, the signature of the metric is Lorentzian, $H$ is accompanied by an $i$ factor, and $Z(t)=\exp(itH)$ is known as the evolution operator.
}

In a general functorial QFT, the partition function for a cobordism is 
\begin{equation}
    Z(\Sigma)\in \mr{Map}(\mr{Geom}(\Sigma),\mr{Hom}(\HH_\mr{in},\HH_\mr{out})),
\end{equation}
where $\mr{Geom}$ is the space of possible geometric data on $\Sigma$. It is understood that $\HH_\mr{in,out}$ can depend on the geometric data restricted to the boundary. So, more appropriately, $Z(\Sigma)$ is a section of the bundle over $\mr{Geom}(\Sigma)$ with fiber over $\gamma\in \mr{Geom}(\Sigma)$ being $\mr{Hom}(\HH_\mr{in}(\gamma|_\mr{in}),\HH_\mr{out}(\gamma|_\mr{out}))$.

In HTQFT,\footnote{H stands for ``Homotopical'' or ``Higher'' or ``chain-level version of coHomological QFT.''} the codomain of the functor is the category of cochain complexes (with differential denoted $Q$); $\mr{Geom}$ is replaced by a differential graded manifold $\mr{sGeom}$ (``s'' for ``super'') with differential $q$. We impose an extra condition of total closedness of the functor
\begin{equation}\label{main eq: (q+Q)Z=0}
    (q+Q)Z=0.
\end{equation}

\subsection{Examples}

\begin{example}%[Strict TQFTs]
    The case of a (strict) TQFT corresponds to $\mr{sGeom}$ being a single point and $Q=0$.
    
    For instance, a one-dimensional TQFT  -- (strict) topological quantum mechanics -- assigns to an interval an element $Z\in \mr{End}(\HH)$ subject to the identity
    \begin{equation}\label{Z^2=Z}
    %\label{U^2=U}
        Z^2=Z,
    \end{equation}
    since the gluing of two intervals is diffeomorphic to a single interval. Note that 
    (\ref{Z^2=Z}) 
    %(\ref{U^2=U})
    means that $Z$ is a projector.
\end{example}

\begin{remark}
    A typical case for a HTQFT is when $\sGeom$ has the form of an odd tangent bundle, 
    \begin{equation}\label{sGeom=T[1]Geom}
    \sGeom=T[1]\Geom
    \end{equation} 
    of some manifold $\Geom$ of geometric data on the cobordism, with $q=d_\Geom$ the de Rham differential. Then the partition function is a nonhomogeneous differential form on $\Geom$,  
    $$Z(\Sigma)\in \Omega^\bt(\Geom)\otimes \mr{Hom}(\HH_\mr{in},\HH_\mr{out}),$$
    with the equation (\ref{main eq: (q+Q)Z=0}) taking the form
\begin{equation}\label{main eq for sGeom=T[1]Geom}
    (d_\Geom+Q)Z=0.
\end{equation}
Denote $Z^{(p)}$, with $p=0,1,2,\ldots$, the component of $Z$ of de Rham degree $p$ along $\Geom$. Then equation (\ref{main eq for sGeom=T[1]Geom}) restricted to de Rham degrees $0$ and $1$ yields
\begin{equation}
    QZ^{(0)}=0,\quad d_\Geom Z^{(0)}=-QZ^{(1)}.
\end{equation}
I.e., (a) $Z^{(0)}$ gives a $Q$-closed state for any geometric data and (b) changing the geometric data infinitesimally changes $Z^{(0)}$ by a $Q$-exact term. In other words, $Z^{(0)}$ defines an element in the cohomology of $Q$, locally constant along $\Geom$. In particular, if $\Geom$ is connected, $Z^{(0)}$ induces a strict TQFT, with spaces of states given by $Q$-cohomology of the spaces of states of the original HTQFT.

We emphasize that there is more information in an HTQFT than in the strict TQFT induced by $Z^{(0)}$ on $Q$-cohomology. For instance, one can integrate a higher form component $Z^{(p)}$ over some distinguished $p$-cycle in $\Geom$ and then project to $Q$-cohomology of the spaces of states. Examples of this construction are 
\begin{itemize}
    \item  Gromov-Witten invariants (arising from the A model) and 
    \item the BV algebra structure on $Q$-cohomology in a 2d topological conformal field theory (see Section \ref{sss: Efr and BV} below).
\end{itemize}
\end{remark}

\begin{comment}
\begin{example}    
    An alternative setup for strict TQFTs is: one may have $Q=0$ and 
    \begin{equation}\label{sGeom = T[1]Geom}
    \mr{sGeom}=T[1]\mr{Geom}
    \end{equation} 
    with $q$ the de Rham differential on a manifold $\mr{Geom}$ (seen as a cohomological vector field on $\mr{sGeom}$) and with $Z$ constant along the tangent fibers of (\ref{sGeom = T[1]Geom}). Then the equation (\ref{main eq: (q+Q)Z=0}) says that $Z$ is a (locally) constant function on $\mr{Geom}$. If $\mr{Geom}$ is connected, this setup is equivalent to the one above, where $\mr{sGeom}$ is a single point.
\end{example}
\end{comment}

\begin{example}[HTQM] 
\label{Ex: HTQM}
In higher topological quantum mechanics, for an interval,\\ ${\mr{sGeom}=T[1]\RR_+}$ with de Rham differential $q=d_t=dt\wedge \frac{d}{dt}$. The space of states for a point is a cochain complex $\HH,Q$
Then the value of the functor on an interval takes the form
$$
Z(t,dt)= \exp\left((d_t+Q)(tG)\right) =\exp(t[Q,G]+dt\, G)
$$
where $G\in \mr{End}(\HH)$ is of degree $-1$ where $G^2=0$. 
%\marginpar{Do we need $G^2=0$? {\bl Probably so, at least we want $[H,G]=0 \Leftrightarrow [Q,G^2]=0$.}}

In particular, the Hamiltonian 
\begin{equation}
H=[Q,G]=QG+GQ
\end{equation} is $\mr{ad}_Q$-exact. 
Also, note that $G^2=0$ implies $[H,G]=0$.

A strict topological quantum mechanics is the case $Q=G=0$.
\end{example}

\subsection{Example: 2d topological conformal field theory}
\label{ss: TCFT}

%\begin{example}
    In 2-dimensional topological conformal field theory (TCFT), the space of local quantum fields (which is the same as the space of states for a tiny circle) $\HH=\HH_{S^1}$ is a cochain complex with differential $Q$ and the stress-energy tensor is $Q$-exact:
    $T=Q(G), \bar{T}=Q(\bar{G})$ for some degree $-1$ fields in $\HH$ (and $T_{z\bar{z}}$ vanishes strictly). 
    Here we assume that $\bar\partial G$ and $\partial \bar{G}$ vanish strictly and that $G,\bar{G}$ have regular OPE \cite[pp.6--7]{Witten93}.
    For a surface $\Sigma$ of genus $h$ with $n$ punctures,\footnote{ One thinks of punctures as tiny disks removed from the surface, each creating an $S^1$ boundary component. We treat all these boundary components as in-boundary.} the space of geometric data is $\sGeom=T[1]\til{\mc{M}}_{h,n}$. Here $\til{\mc{M}}_{h,n}$ is the moduli space of complex structures (non-compactified, with ordered punctures),
    equipped with a formal coordinate chart on $\Sigma$ at each puncture, see \cite{Zwiebach}. We will refer to $\til{\mc{M}}_{h,n}$ as \emph{Zwiebach's moduli space}.
    Thus, the partition function 
    \begin{equation}\label{TCFT: Z is a form on M_h,n}
    Z(\Sigma)\in \Omega^\bt(\til{\mc{M}}_{h,n},\mr{Hom}(\HH^{\otimes n},\CC)) 
    \end{equation}
    is a nonhomogeneous differential form on $\til{\mc{M}}_{h,n}$ with values in $\mr{Hom}(\HH^{\otimes n},\CC)$.\footnote{
    For the purpose of this section, since we are talking about differential forms, we have to fix the ground field to be $\kk=\CC$ (or $\RR$).
    } The $p$-form component of $Z$ is defined as a correlator involving $p$ insertions of the field $G^\mr{tot}=G(dz)^2+\bar{G}(d\bar{z})^2 $:
    \begin{equation}\label{TCFT Z}
        Z^{(p)}(\Phi_1\otimes\cdots \otimes \Phi_n) = \Big\langle \underbrace{G^\mr{tot}\cdots G^\mr{tot}}_p \Phi_1(z_1)\cdots \Phi_n(z_n) \Big\rangle_\Sigma.
    \end{equation}
    With $\Phi_1,\ldots,\Phi_n\in \HH$ a collection of fields (observables\footnote{Here ``observables'' does not imply that $\Phi_i$ are $Q$-closed.}) plugged in the punctures $z_1,\ldots,z_n$. More explicitly, if we have $p$ tangent vectors to $\til{\MM}_{h,n}$ at a given complex structure, described by Beltrami differentials $\mu_i+\bar\mu_i$, then the contraction of (\ref{TCFT Z}) with these tangent vectors is
     \begin{equation}
        \iota_{\mu_1+\bar\mu_1}\cdots \iota_{\mu_p+\bar\mu_p} Z^{(p)}(\Phi_1\otimes\cdots \otimes \Phi_n) = \Big\langle \prod_{i=1}^p \int_\Sigma \iota_{\mu_i+\bar\mu_i} G^\mr{tot} %+\bar\mu_i \bar{G})
        \cdot \Phi_1(z_1)\cdots \Phi_n(z_n) \Big\rangle_\Sigma.
    \end{equation}
    Equation (\ref{main eq: (q+Q)Z=0}) in this case has the form %\marginpar{sign}
    \begin{equation}\label{TCFT closedness}
        d_\MM Z^{(p-1)}(\Phi_1\otimes\cdots\otimes\Phi_n)+\sum_{j=1}^n \pm Z^{(p)}(\Phi_1\otimes \cdots \otimes Q(\Phi_j)\otimes \cdots \otimes \Phi_n)=0.
    \end{equation}
    In the left term $d_\MM$ is the de Rham differential on the moduli of complex structures. In that term, the action of $d_\MM$ is tantamount to replacing one of $G^\mr{tot}$ insertions by the stress-energy tensor $T^\mr{tot}$. The equation (\ref{TCFT closedness}) is standard and is proven by a contour-trading trick.

    An important special case: if the field insertions $\Phi_i$ are $Q$-closed, (\ref{TCFT Z}) is a closed form on the moduli space. This restriction induces a reduced version of a TCFT, assigning to a surface a de Rham cohomology class of the moduli space
    \begin{equation}\label{CohFT Z}
        Z_\mr{red}\in H^\bt%_\mr{de\, Rham}
        (\til{\MM}_{h,n})\otimes \mr{Hom}(\HH_\mr{red}^{\otimes n},\CC),
    \end{equation}
    where 
    \begin{equation}\label{Q-cohomology}
    \HH_\mr{red}\colon=H_Q(\HH)
    \end{equation} 
    is the $Q$-cohomology of the space of states/local quantum fields. 
    
   In a class of TCFTs (e.g. in the A model) the correlators of some $Q$-closed observables (``evaluation observables'' in the case of the A model, see \cite{FLN}) descend from $\til\MM_{h,n}$ to $\MM_{h,n}$ and then extend to the Deligne-Mumford compactification $\ol\MM_{h,n}$ of the moduli space $\MM_{h,n}$. 
    In this case, in (\ref{CohFT Z}) one can replace $H^\bt(\til{\MM}_{h,n})$ with $H^\bt(\ol\MM_{h,n})$.
    In this case, 
    by virtue of cutting-gluing property (functoriality) of the original TCFT, the reduced partition functions (\ref{CohFT Z}) factorize on nodal surfaces, which results in the famous WDVV equation. This reduction is known as ``cohomological field theory,'' see \cite{KM}, also \cite{LS}. The value of this cohomology class paired with fundamental class of $\ol{\MM}_{h,n}$ is known as the ``$h$-loop string amplitude.''
%\end{example}

%[ON THE FRAMED E2 OPERAD]
\subsubsection{TCFT in genus zero: framed little disk operad and the BV algebra structure}
\label{sss: Efr and BV}

Recall (see \cite{Getzler} for details) that the framed operad of little 2-disks $E_2^\fr$ has as its $n$-th space $E_2^\fr(n)$ the configuration space of $n\geq 0$ disjoint disks (or ``holes'') inside the standard unit disk in $\CC$, where all disks are equipped with a marked point on the boundary circle; the outer circle has the standard marked point at $1$. The operadic compositions $\circ_i\colon E_2^\fr(n)\times E_2^\fr(m)\ra E_2^\fr(n+m-1)$ are given by scaling and rotating one disk configuration and replacing with it the $i$-th disk of the second configuration (Figure \ref{fig: E2^fr composition}).
\begin{figure}[h]
    \centering
    \includegraphics[scale=0.6]{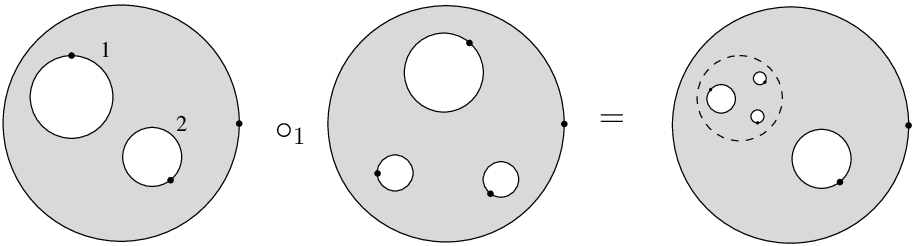}
    \caption{Composition in the framed little disk operad.}
    \label{fig: E2^fr composition}
\end{figure}

Homology of $E_2^\fr$ is generated (via operadic composition) by the following classes (Figure \ref{fig: E2^fr homology}):
\begin{enumerate}[(i)]
    \item %Homology class of the 
    Unit\footnote{ By abuse of notations, here we use the same notations and names for generators of homology of $E_2^\fr$ as those commonly used for the BV algebra operations representing these homology classes on $H_Q(\HH)$.
    } $1$ -- the generator of $H_0(E_2^\fr(0))$. %(inner disk ``sitting still'').
    \item %Homology class of the 
    Product $\cdot \in H_0(E_2^\fr(2))$ (two inner disks ``sitting still'').
    \item %Homology class of the 
    BV operator $\Delta \in H_1(E_2^\fr(1))$ -- corresponds to the inner circle making a full rotation. We denote the BV operator $\Delta$ or $G_{0,-}$.
    \item %Homology class of the 
    BV bracket (a.k.a. odd Poisson bracket) $\{,\}\in H_1(E_2^\fr(2))$ -- corresponds to one inner disks moving a full circle around the other inner disk.
\end{enumerate}
\begin{figure}[h]
    \centering
    \includegraphics[scale=0.6]{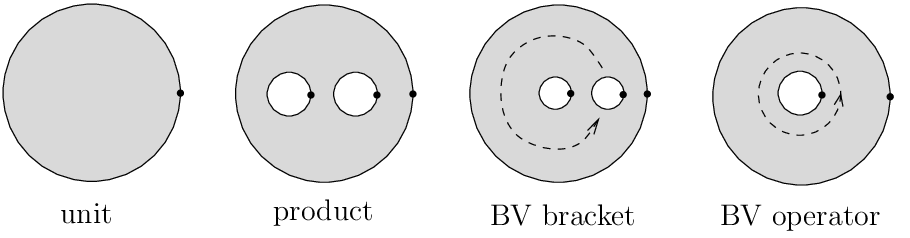}
    \caption{Generators of homology of the framed little disk operad.}
    \label{fig: E2^fr homology}
\end{figure}

These generators satisfy the following relations:
\begin{enumerate}[(a)]
\item The product is supercommutative and associative with $1$ the unit; $\Delta(1)=0$.
    \item BV operator squares to zero: 
    \begin{equation}
    \Delta^2=0.
    \end{equation}
    \item $\Delta$ is a second-order derivation of the product: 
    \begin{multline}\label{7-term relation}
    \Delta(xyz)-(-1)^{|x|}x\Delta(yz)-(-1)^{|y||z|}\Delta(xz)y-\Delta(xy)z+\\
    +(-1)^{|x|+|y|} xy\Delta(z)+(-1)^{|x|}x\Delta(y)z+\Delta(x)yz=0
    \end{multline}
    -- the ``7-term relation.'' %\marginpar{+ Koszul signs}
    \item The BV bracket is the ``defect'' of the first-order Leibniz identity for $\Delta$:
    \begin{equation}\label{Delta Leibniz}
    \Delta(xy)=\Delta(x)y+(-1)^{|x|} x\Delta(y)+(-1)^{|x|}\{x,y\}.
    \end{equation}
    (As a consequence, $\{,\}$ is a biderivation of the product and satisfies Jacobi; $\Delta$ is a derivation of $\{,\}$.)
\end{enumerate}
Homology of $E_2^\fr$ is known as the BV (Batalin-Vilkovisky) operad.

Restricting a TCFT to genus zero cobordisms -- disks with holes, with the outer circle seen as ``out'' and inner circles seen as ``ins''  -- yields a sequence of differential forms on the configuration space of disks,
\begin{equation}\label{TCFT: Z as a form on E_2^fr}
    Z_{[n]}\in \Omega^\bt(E_2^\fr(n))\otimes \mr{Hom}(\HH^{\otimes n},\HH)
\end{equation}
satisfying the equation (\ref{main eq: (q+Q)Z=0}), which in this case becomes
\begin{equation}\label{framed little disks: main eq}
    (d_{E_2^\fr}+Q) Z_{[n]} =0
\end{equation}
with $d_{E_2^\fr}$ the de Rham operator on the configurations space of disks. Equation (\ref{framed little disks: main eq}) together with functoriality of the original field theory implies that the space of states for the circle $\HH$  carries a dg representation of the operad $E_2^\fr$. In other words, $\HH$ is an $E_2^\fr$-algebra. 

As a consequence, the reduced space of states -- the $Q$-cohomology $\HH_\mr{red}=H_Q(\HH)$
%(\ref{Q-cohomology}) 
-- has the structure of a BV algebra.

\begin{remark}
%Remark: 
For a TCFT where the correlators of $Q$-closed observables extend to the Deligne-Mumford compactification of the moduli space $\ol\MM_{h,n}$, the BV algebra on $H_Q(\HH)$ has vanishing BV operator (and hence also the BV bracket, cf. (\ref{Delta Leibniz})). Thus, the BV operator plays the role of an obstruction for extendability of the correlators to the Deligne-Mumford compactified moduli space. See e.g. \cite{LMY1} for an explicit example of a TCFT with nontrivial BV algebra on $Q$-cohomology. %\marginpar{mention extendability to the Kimura-Stasheff-Voronov compactification? }
\end{remark}

\begin{remark}
    To  shed more light on %see more explicitly 
    the relation between the TCFT as a form on the moduli space $\til{\MM}_{h,n}$ (\ref{TCFT: Z is a form on M_h,n}) and genus zero TCFT as a form on the framed little disk operad (\ref{TCFT: Z as a form on E_2^fr}), we note the following.  The moduli space $\til\MM_{0,n+1}$ contracts onto $\MM'_{0,n+1}$ -- a $(\CC^*)^{n+1}$-bundle over $\MM_{0,n+1}$ (since the group of formal pointed diffeomorphisms of the tangent space to $\Sigma$ at a puncture contracts onto $\CC^*$), which in turn is homotopy equivalent to $E_2^\fr(n)$.\footnote{
More precisely, one has a homeomorphism $\MM'_{0,n+1}\simeq E_2^\fr(n)/\gamma$ where $\gamma$ is the subgroup of the group of M\"obius transformations $PSL_2(\CC)$ mapping the unit circle to itself and the preserving the point $1$. The subgroup $\gamma$ is conjugate to the subgroup $\{z\mapsto az+b\; |\; a>0, b\in \RR\}$ -- conformal automorphisms of the upper half-plane.
    }
\end{remark}

\subsection{HTQFT on triangulated cobordisms -- a preliminary discussion}
We would like to consider the triangulations of a cobordism $\Sigma$ as vertices of a certain CW complex $\Xi(\Sigma)$  (we will refer to it as the ``Pachner complex''), with edges corresponding to Pachner moves between triangulations. We will discuss higher cells later. In this setting, the partition function is a cellular cochain on $\Xi$ with coefficients in $\HH$, $Z(\Sigma)\in C^\bt(\Xi,\HH)$ that should satisfy the equation
\begin{equation}\label{triang cob main eq}
(\delta_\Xi+Q)Z=0,
\end{equation} 
 with $\delta_\Xi$ the coboundary operator on cochains. Thus,  here the Pachner complex $\Xi$ plays the role of $\sGeom$, and cellular cochains on $\Xi$ play the role of functions on $\sGeom$ (or equivalently differential forms on $\Geom$, cf. (\ref{sGeom=T[1]Geom})).
Equation (\ref{triang cob main eq}) is a cellular analog of equation (\ref{main eq for sGeom=T[1]Geom}).

Denote $Z^{(p)}$ the degree $p$ component of $Z$ as a cochain on $\Xi$. Then equation (\ref{triang cob main eq}) says
\begin{equation}
    \delta_\Xi Z^{(p-1)}+Q Z^{(p)}=0. 
\end{equation}
In particular, for $p=0$ and $p=1$ we have %that $Z^{(0)}$ is $Q$-closed.
\begin{equation}
    Q Z^{(0)}=0,\quad \delta_\Xi (Z^{(0)})=-Q Z^{(1)}.
\end{equation}
I.e., the partition function on a triangulated cobordism is $Q$-closed and changes by $Q$-exact terms under Pachner moves.
%It is instructive to examine the equation (\ref{triang cob main eq}) in low cochain degrees on $\Xi$. 

%\section{
%Two-dimensional HTQFT on triangulated cobordisms from 
%A combinatorial two-dimensional HTQFT
%}

\section{
%First attempt at 2d HTQFT on triangulated cobordisms: flip complex and cyclic $A_\infty$ algebra
%First attempt at 2d HTQFT on triangulated cobordisms: HTQFT on the flip complex from a cyclic $A_\infty$ algebra
First attempt at 2d HTQFT on triangulated cobordisms: HTQFT on the flip complex from a cyclic $A_\infty$ algebra (``Model 1'')
}
\label{sec: flip theory}

%\marginpar{This is quite long. Can move parts of it to the intro?}
Our strategy now is to construct a 2d HTQFT on triangulated cobordisms by relaxing the construction of Section \ref{sec: combinatorial strict 2d TQFT}: we will require the associativity of the product $m_2$ in $V$ and the relation (\ref{Pachner 2 rel (strict)}) to hold only up to homotopy (i.e. up to a $Q$-exact term $Q(\cdots)$, where the homotopy $\cdots$ will be taken as a part of the structure), ensuring that invariance with respect to Pachner moves holds up to homotopy. In this section we will focus only on Pachner flips and relaxing associativity up to homotopy. 
In Section \ref{sec: full 2d HTQFT on triang cob} we will incorporate the stellar subdivisions/aggregations and invariance up to homotopy with respect to them. 

A part of the discussion in both cases is the appropriate treatment of higher homotopies. In the case of flips, this results in having a homotopy for associativity assigned to squares and higher homotopies assigned to polygons with $\geq 5$ vertices, 
with the overall local structure being that of a cyclic $A_\infty$ algebra (replacing an associative algebra with an inner product, in the setup of Section \ref{sec: combinatorial strict 2d TQFT}). In the case of stellar subdivisions/aggregations, that structure needs to be enhanced with further homotopies associated with polygons with ``floating vertices'' inside; this new  structure comes from the construction of ``secondary polytopes'' by Gelfand-Kapranov-Zelevinsky \cite{GKZ}, see also \cite{KKS}.

\subsection{Local data: cyclic $A_\infty$ algebra}\label{ss: local data: cyclic A_infty algebra}
Let $V$ be a $\mathbb{Z}$-graded vector space over $\kk$ endowed with the structure of an $A_\infty$ algebra\footnote{See e.g. \cite{MSS} and \cite{LV}, see also \cite{Kadeishvili}.}
%\marginpar{cite Stasheff? Keller?}
-- a sequence of multilinear operations $m_n\colon V^{\otimes n}\ra V$ of degree $2-n$, 
with $n=1,2,3\ldots$ Here $Q=m_1$ is a differential\footnote{\label{footnote 16}
By abuse of notations, we also denote by $Q$ the induced (via Leibniz identity) differential on tensors -- elements of $\mr{Hom}(V^{\otimes k},V^{\otimes l})$.
 For instance in (\ref{A_infty rel 2}), $m_2$ is an element of $\Hom(V^{\otimes 2},V)$ and so $Q(m_2)\in \mr{Hom}(V^{\otimes 2},V)$ is defined by $(Q(m_2))(x,y)={Q(m_2(x,y))-m_2(Qx,y)-(-1)^{|x|}m_2(x,Qy)}$. 
} and the operations are understood to satisfy $A_\infty$ relations:
\begin{eqnarray}
    Q^2&=&0 \label{A_infty rel 1}\\
    %Qm_2(x,y)-m_2(Qx,y)-m_2(x,Qy) 
    {Q(m_2)}
    &=& 0  \label{A_infty rel 2}\\
    m_2(m_2(x,y),z)-m_2(x,m_2(y,z)) &=& -Q(m_3)(x,y,z) \label{A_infty rel 3}\\
    &\vdots & \nonumber \\
    {Q(m_n)}&=&\sum_{k+l=n+1,\;k,l\geq 2} m_k\circ m_l
    %[m_k,m_l]_G 
    \label{A_infty rel n}
    \\ &\vdots & \nonumber
\end{eqnarray}
where  $m_k\circ m_l = \sum_{i=1}^k \pm m_k \circ_i m_l$ and $\circ_i$ is the $i$-th composition, putting the output of $m_l$ as $i$-th input of $m_k$.\footnote{
Another description: one can extend a multilinear map $A\in \Hom(V^{\otimes k},V)$ to a coderivation of the tensor coalgebra of $V[1]$,  $\widehat{A}\in \mr{Coder}(\bigoplus_{i\geq 1} (V[1])^{\otimes i})$. Then the composition $A\circ B$ is defined via $\widehat{A\circ B}=\widehat{A}\circ \widehat{B}$, where on the right one has the ordinary composition of coderivations. In the language of coderivations, the $A_\infty$ relations jointly are equivalent to the equation ${(\wh{m}_1+\wh{m}_2+\cdots)^2=0}$.
}
%$[,]_G$ is the Gerstenhaber bracket in Hochschild cochains $HC^\bt(V,V)$. 
In particular, 
\begin{itemize}
    \item (\ref{A_infty rel 1}) says that $Q$ is a differential,
    \item (\ref{A_infty rel 2}) is Leibniz identity for $m_2$,
    \item (\ref{A_infty rel 3}) is associativity of $m_2$ up to homotopy, with $m_3$ being the homotopy,
    \item subsequent relations (\ref{A_infty rel n}) are coherencies on higher homotopies.
\end{itemize}

We will assume additionally that $V$ is equipped with a nondegenerate symmetric pairing $g\colon V\otimes V\ra \kk$ of degree zero, such that operations
\begin{equation}\label{c_n+1}
\begin{array}{cccc}
    c_{n+1}=g\circ (\mr{id}\otimes m_n)\colon &V^{\otimes (n+1)} &\ra & \kk  \\
     &   x_0\otimes x_1 \ldots \otimes x_{n} & \mapsto & (-1)^{n|x_0|} g(x_0,m_n(x_1,\ldots,x_n))
\end{array}
\end{equation}
are invariant under cyclic permutations of inputs for $n=1,2,\ldots$.\footnote{
With explicit signs: 
$c_{n+1}(x_r,\ldots,x_n,x_0,\ldots,x_{r-1})=(-1)^{\sigma
}c_{n+1}(x_0,\ldots,x_{r-1},x_r ,\ldots,x_n)$ with
$\sigma= nr+(\sum_{i=r}^n |x_i|)(\sum_{j=0}^{r-1}|x_j|)$,
for $r=0,\ldots,n$.
} We will refer to $c_{n+1}$ as ``cyclic operations.'' The data of an $A_\infty$ algebra $(V;m_1,m_2,\ldots)$ together with such pairing $g$ is known as a ``cyclic $A_\infty$ algebra.''\footnote{See e.g. \cite{Tradler}.}

\subsection{%Flip complex and the partition function
Flip complex
}
Let $\Sigma$ be an oriented surface, possibly with boundary. Let $P$ %\marginpar{rename P to S}
be a finite collection of points (``vertices'') in $\Sigma$, either in the bulk or on the boundary. We assume that each boundary circle contains at least one point of $P$. We consider polygonal decompositions $\alpha$ (cell decompositions with 2-cells being $n$-gons with $n\geq 3$) of $\Sigma$ with vertices at $P$, up to cellular homeomorphism  relative to $P$ and relative to the boundary.\footnote{We do not require the polygonal decompositions to be \emph{regular} CW complexes: edges or vertices of a closure of a cell are allowed to be glued together.}
%\marginpar{Relative to $\dd\Sigma$? (but this seems to be subsumed by rel $P$ condition)}
%{\color{red} [REL BOUNDARY?]}.

We consider the ``flip complex'' $\Xi_\flip(\Sigma,P)$ -- a CW complex where:
\begin{itemize}
    \item $0$-cells are associated with triangulations of $\Sigma$ with vertices at $P$.
    \item A polygonal decomposition $\alpha$ of $\Sigma$ with vertices at $P$ whose $2$-cells are $n_i$-gons, is assigned a cell $e_\alpha$ in $\Xi_\flip$ of dimension $\sum_i(n_i-3)$. This cell 
    %{\red [OR RATHER ITS CLOSURE? (THOUGH THIS DOESN'T WORK WITH SIMPLE CYLINDER CASE)]} 
    can be identified with (the interior of) a product of Stasheff's associahedra
    $\prod_i K_{n_i-1}$ (Figure \ref{fig: K2-K5}).
    %\footnote{
    %Recall that, e.g., $K_2$ is a point, $K_3$ is an interval, $K_4$ is a pentagon, $K_5$ is 3-dimensional polyhedron  (with $9$ faces -- three squares and six pentagons), etc.
    %}
    %\item 
    The boundary of this cell %$c_\alpha$ 
    is the union of cells one dimension lower corresponding to subdividing one polygon in $\alpha$ by some diagonal into two polygons (they should both be at least triangles).
\end{itemize}
\begin{figure}[h]
    \centering
    \includegraphics[scale=0.7]{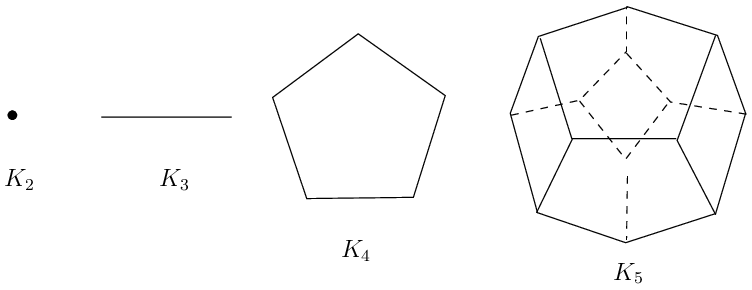}
    \caption{First few associahedra.}
    \label{fig: K2-K5}
\end{figure}

\begin{example}
1-cells of $\Xi_\flip$ correspond to triangulations of $\Sigma$ with one square. The boundary of this 1-cell is the difference of two triangulations obtained by cutting that square by a diagonal in two possible ways, see Figure \ref{fig: 1-cell of Xi_flip}. 
\begin{figure}[h]
    \centering
    \includegraphics[scale=0.7]{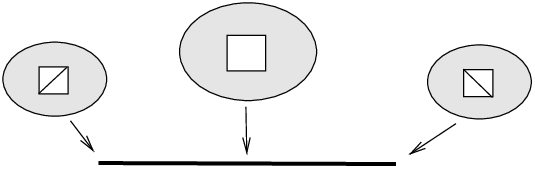}
    \vspace{0.3cm}
    \caption{1-cell $e$ in $\Xi_\flip$. We show the polygonal decompositions of $\Sigma$ decorating the strata (sub-cells of $\bar{e}$). Shaded region stands for the triangulated part of $\Sigma$.}
    \label{fig: 1-cell of Xi_flip}
\end{figure}

2-cells of $\Xi_\flip$ correspond to either (a) triangulations with two squares, or (b) triangulations with one pentagon. Coincidentally, in the case (a) the cell in $\Xi_\flip$ itself looks like a square (its boundary corresponds to splitting one of the two squares in $\alpha$ by a diagonal), and in the case (b) the corresponding cell in $\Xi_\flip$ looks like a pentagon, see Figure \ref{fig: 2-cells in Xi_flip}.
\begin{figure}[h]
    \centering
    $$\vcenter{\hbox{
    \includegraphics[scale=0.65]{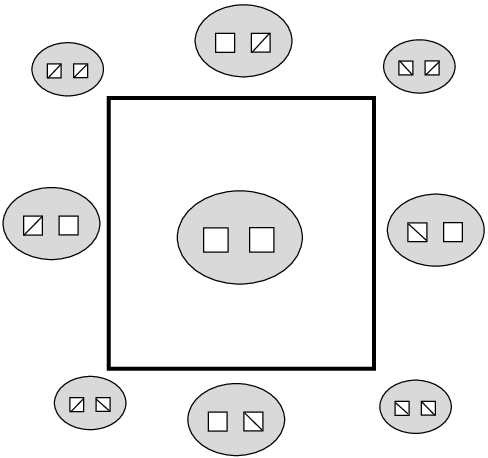} \qquad
    \includegraphics[scale=0.6]{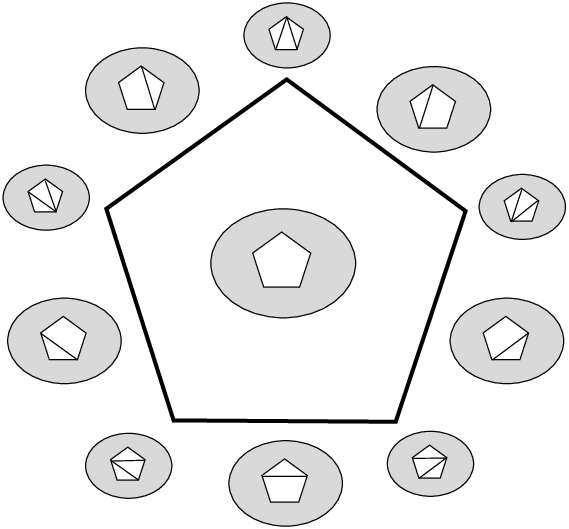}}}$$
    \caption{2-cells in $\Xi_\flip$.
    %2-cells $e$ in $\Xi_\flip$. We show the polygonal decompositions of $\Sigma$ decorating the strata (sub-cells of $\bar{e}$). Shaded region stands for the triangulated part of $\Sigma$.
    }
    \label{fig: 2-cells in Xi_flip}
\end{figure}
\end{example}

\begin{example}
    For $\Sigma$ a disk with $P$ an $n$-tuple of points on the boundary, $\Xi_\flip$ is the associahedron $K_{n-1}$.
\end{example}

\begin{example}\label{ex: Xi flip for simple cylinder}
    For $\Sigma$ a cylinder with $P$ consisting of two points -- one on each boundary circle -- $\Xi_\flip$ is homeomorphic to a circle, see Figure \ref{fig: Xi flip for simple cylinder}.
    \begin{figure}[h]
        \centering
        \includegraphics[scale=0.5]{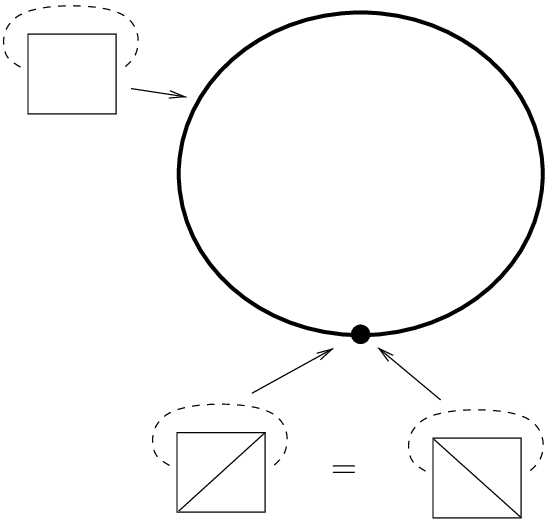}
        \caption{$\Xi_\flip$ for a cylinder with one vertex on each boundary. Dashed line means that the two sides of a square are identified (to obtain the cylinder). Note that 
        the two triangulated squares with identified opposite sides are cellular homeomorphic (related by Dehn twist) and hence correspond to the same vertex of $\Xi_\flip$.
        }
        \label{fig: Xi flip for simple cylinder}
    \end{figure}
\end{example}

\begin{example}
    For $\Sigma$ a disk with $P$ consisting of three points on the boundary and one point in the bulk, $\Xi_\flip$ is itself homeomorphic to a disk, with cell decomposition into three pentagons, see Figure \ref{fig: 3 pentagons}.
    \begin{figure}[h]
        \centering
        \includegraphics[scale=0.7]{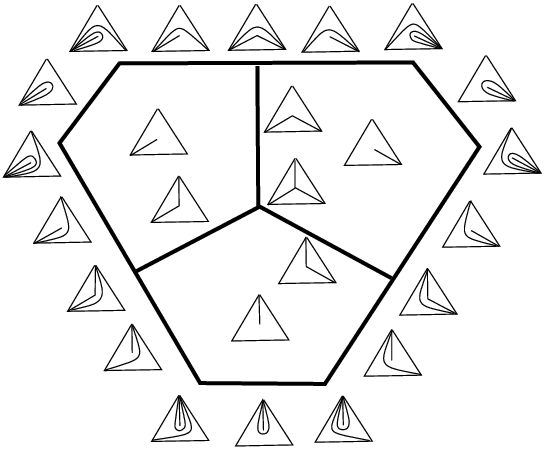}
        \caption{$\Xi_\flip$ for a triangle with an extra vertex in the bulk. Note that ``singular'' polygons (with identified edges) occur abundantly in this example.  E.g. the 2-cells of $\Xi_\flip$ correspond to representations of $\Sigma$ as a pentagon with glued consecutive edges.}
        \label{fig: 3 pentagons}
    \end{figure}
Note that one can also has conventional ribbon graphs dual to the polygonal decompositions of $\Sigma$ with singular polygons, see e.g. Figure \ref{fig: droplet}.
%\marginpar{new}
\begin{figure}[h]
    \centering
    \includegraphics[scale=0.4]{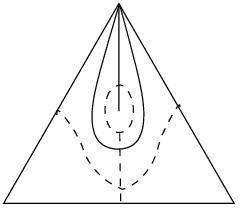}
    \caption{A polygonal decomposition with a singular polygon and the dual ribbon graph (shown in dashed lines).}
    \label{fig: droplet}
\end{figure}
    
\end{example}

%{\color{green} Relation of $\Xi_\flip$ (for $\Sigma$ closed) (via Poincar\'e duality?) to moduli space of cx structures, via moduli of metric ribbon graphs, cite Mulase-Penkava.}

\subsubsection{Aside: polygonal decompositions, ribbon graphs and the moduli space $\MM_{h,N}$}
\label{sss: polygonal decompositions, ribbon graphs and the moduli space}

%A note on notations: we use $N$ for the number of vertices in $P$ while $n$ is used for the number of holes in the surface. 

In this section we compare $\Xi_\flip$ with the moduli space $\MM_{h,N}$. Here, crucially,  $N$ is the number of vertices in $P$, \emph{not} the number of holes in the surface $\Sigma$ (which we denote $n$).
On the other hand, the relevant moduli space for comparison with TCFT is $\MM_{h,n}$, see Section \ref{sss: analogy between comb HTQFT and cont TCFT} below. 
I.e., in this section vertices of polygonal decompositions correspond to marked points on a Riemann surface, while in comparison to field theory, polygonal holes become marked points/punctures.

    For a polygonal decomposition $\alpha$ of a surface $\Sigma$ with $P$ the set of vertices, one has the dual ribbon graph $\Gamma$ on $\Sigma$ with $P$ the set of ``borders" or ``faces." If $\Sigma$ has boundary, then $\Gamma$ is a ribbon graph with ``leaves'' (loose half-edges), dual to the boundary edges of $\alpha$.

    If $\Sigma$ is a closed surface of genus $h$ and $P$ is a set of $N$ points, one has a natural bijection $\rho$ between cells $e_\alpha$ of $\Xi_\flip$ and building blocks (``orbi-cells''\footnote{
    $RG^\mr{met}$ is a noncompact orbi-cell complex, and $C_\Gamma$ are its strata -- the noncompact orbi-cells.
    }) 
    \begin{equation}
    C_\Gamma=\RR_+^{E(\Gamma)}/\mr{Aut}(\Gamma)
    \end{equation} 
    of the space of metric ribbon graphs $RG^\mr{met}_{h,N}$, with $\Gamma$ the ribbon graph dual to $\alpha$. The latter is the well-known combinatorial model for the ``decorated'' moduli space of complex structures\footnote{
    We are considering the non-compactified moduli space with ordered marked points.
    } $\MM_{h,N}\times \RR_+^N$ due to Harer-Mumford-Thurston and Penner, see \cite{Harer,Penner,Kontsevich,Mulase}. 
    
    For the rest of this subsection we will assume that $\mr{Aut}(\Gamma)=1$ for all ribbon graphs involved; it is true if $h=0, N\geq 3$ or if $h>0$ and $N$ is large enough.\footnote{ E.g., for genus $h=1$ one needs $N\geq 5$ to kill the automorphisms of ribbon graphs. As an evidence, note that the orbifold Euler characteristic $\chi(\MM_{1,N})=(-1)^N \frac{(N-1)!}{12}$ (cf. \cite{Kontsevich}) is an integer for $N\geq 5$.}
%    \footnote{We could avoid making this assumption, but then we would need to define $\Xi_\flip$ as an orbi-cell complex.} 
%    \marginpar{True?} 
    
    The bijection $\rho$ is compatible with attaching maps but reverses the incidence and reverses
    %(with a shift) 
    the dimensions of cells:
    \begin{itemize}
        \item A triangulation $\alpha$ corresponds to a $0$-cell in $\Xi_\flip$. The corresponding ribbon graph $\Gamma$ is trivalent and has 
        \begin{equation}
        E=\underbrace{6h-6+3N}_{=\colon D} = \dim (\MM_{h,N}\times \RR_+^N)
        \end{equation} edges and gives a top-dimension building block of $RG^\mr{met}$.
        \item For $\alpha$ a general polygonal decomposition, the corresponding ribbon graph  contains vertices of valence $\geq 4$ (corresponding to squares, pentagons etc. in $\alpha$) and has $E=D-  \dim e_\alpha$ edges.
        \item The lowest dimension blocks of $RG^\mr{met}$ correspond to graphs $\Gamma$ with a single vertex and $E=2h-1+N$ edges. They correspond to top-dimension cells of $\Xi_\flip$, with $\alpha$ containing a single polygon; the dimension of these cells in $\Xi_\flip$ is $\dim e_\alpha=4h-5+2N$.
    \end{itemize}
    In particular, for $N$ sufficiently large, one has the following.
    \begin{itemize}
    \item $\Xi_\flip(\Sigma,P)$ and $RG^\mr{met}$ are Poincar\'e dual cell complexes (one of them noncompact), see Figure \ref{fig: M_0,3}.
    %therefore   $\Xi_\flip(\Sigma,P)$ is homotopy equivalent to $\MM_{h,n}$. 
    \begin{figure}[h]
    \centering
    \includegraphics[scale=0.9]{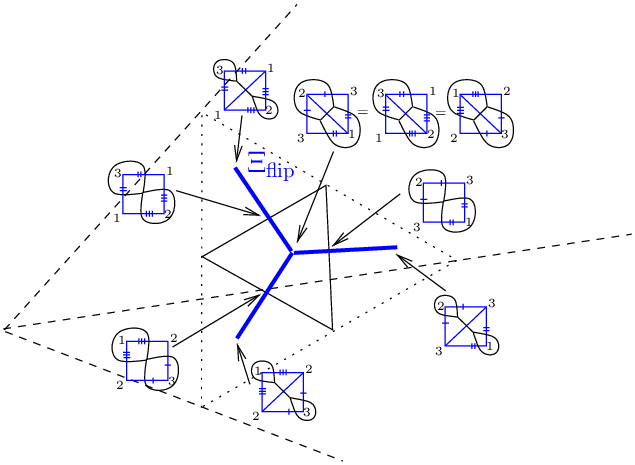}
    \caption{Stratification of $\MM_{0,3}\times \RR_+^3$ colored by ribbon graphs (shown as a stratification of the cross-section of the octant) and the corresponding $\Xi_\flip$ (in blue). Numbers are labels of vertices (resp., borders of ribbon graphs). Gluing of edges in polygonal decomposition is shown by little dashes.}
    \label{fig: M_0,3}
\end{figure}
    \item  One has a canonical isomorphism $\rho_*$ between $p$-chains of $\Xi_\flip$ and $(D-p)$-cochains of $RG^\mr{met}$ for any $p$, acting on generators by 
    \begin{equation}
    \rho_*(e_\alpha)(C_\Gamma)= \left\{ 
    \begin{array}{cc}
         1& \mbox{if $\Gamma$ is dual to $\alpha$,}  \\
         0& \mbox{otherwise} 
    \end{array}
    \right.
    \end{equation}
    %\marginpar{correct?}
    \item One has an isomorphism between homology of $\Xi_\flip$ and compactly supported cohomology (and, via Poincar\'e duality, % Borel-Moore 
    homology)  of the moduli space:
    \begin{equation}
        H_p(\Xi_\flip)\stackrel{\rho_*}{\simeq} H^{D-p}_{c}(\MM_{h,N}) \simeq %H_{\dim \MM_{h,N}-p}^\mr{BM}(\MM_{h,N}).
        H_p(\MM_{h,N})
    \end{equation}
    \end{itemize}

    An important consequence is that adding extra vertices in $P$ makes the homotopy type of $\Xi_\flip$ more complicated, in particular increases %its Betti numbers.
    the rank of first homology.
    %\marginpar{EDIT?} 
    E.g., for $h=0$ ($\Sigma$ a sphere), $N=3$ points, $\Xi_\flip$ is contractible, while for $N=4$ it has nontrivial $H_1$:  $H_1(\Xi_\flip)\simeq H_1(\MM_{0,4})\simeq \ZZ^2$.

    %\textbf{Question.} 
    \begin{question} %\label{question: Xi_flip vs Ribbon graphs}
   % \marginpar{Can put it in the end, in an ``Open questions'' section?}
 %   \begin{enumerate}[(a)]
%    \item  
    \label{question: Xi_flip vs Ribbon graphs. (a) homotopy equivalence} Does cellular Poincar\'e duality above give a homotopy equivalence $\Xi_\flip\sim \MM_{h,N}$?
\end{question}
%    \item 
\begin{question}
    \label{question: Xi_flip vs Ribbon graphs. (b) boundary}
    Does the discussion above extend to surfaces with boundary? I.e., does one have a ribbon graph model for the moduli space of complex structures on a surface with boundary, with marked points in the bulk and on boundary circles, and does one have a comparison of homology of $\Xi_\flip$ with (co)homology of the moduli space going through the appropriate ribbon graph complex?
%    \end{enumerate}
    \end{question}

\subsection{HTQFT on the flip complex} %\marginpar{I made it a separate subsection. Is it OK?}

We define our HTQFT (we will refer to it as ``flip theory'') as follows. 
%the partition function of 
Given a surface $\Sigma$ with set of vertices $P$,
the space of states for $j$-th boundary circle  is
\begin{equation}
    \HH(S^1_j,P\cap{S^1_j})=V^{\otimes k_j}
\end{equation}
where $k_j$ is the number of points of $P$ on $S^1_j$; we understand factors $V$ as associated to the \emph{arcs} into which these points divide the circle.

The partition function
is a cellular cochain on $\Xi_\flip$,
\begin{equation}
    Z\in C^\bt(\Xi_\flip)\otimes \Hom(\HH_\mr{in},\HH_\mr{out})
    %\bigotimes_j \HH(S^1_j,P|_{S^1_j})
\end{equation}
defined as follows: its value on the cell $e_\alpha$ of $\Xi_\flip$ corresponding to a polygonal decomposition $\alpha$ of $\Sigma$ with vertices at $P$ is
\begin{equation}\label{Z(e alpha)}
    Z(e_\alpha)= \Big\langle \bigotimes_{\mr{polygons}\,p\,\mr{of}\,\alpha} c_{|p|}, \bigotimes_{\mr{edges}\setminus \{\mr{in-edges}\}} g^{-1}   \Big\rangle.
\end{equation}
Here $|p|$ is the number of edges of the polygon $p$; $c_{\cdots}$ are the cyclic $A_\infty$ operations (\ref{c_n+1}). We will also denote (\ref{Z(e alpha)}) by $Z(\alpha)$. Note that formula (\ref{Z(e alpha)}) is a natural generalization of the strict example (\ref{Z strict}), allowing for polygons that are not triangles and decorated by higher $A_\infty$ operations.

\begin{thm}\label{thm: model 1}
Properties of the construction:
\begin{enumerate}[(i)]
    \item Functoriality\footnote{Cf. Section \ref{ss: functoriality of strict tqft}.} (compatibility with gluing): given two composable cobordisms endowed with polygonal decompositions
    \begin{equation}
        (\gamma_1,\tau_1) \xra{(\Sigma',\alpha')} (\gamma_2,\tau_2) \xra{(\Sigma'',\alpha'')} (\gamma_3,\tau_3)
    \end{equation}
    (with $\gamma_i$ collections of circles and $\tau_i$ the triangulations of boundaries induced by the polygonal decompositions $\alpha',\alpha''$ of surfaces),
    for their composition 
    \begin{equation}
        (\gamma_1,\tau_1) \xra{(\Sigma,\alpha)} (\gamma_3,\tau_3) 
    \end{equation}
    one has 
    \begin{equation}\label{flip HTQFT functoriality}
        Z(\Sigma,\alpha)=Z(\Sigma'',\alpha'')\circ Z(\Sigma',\alpha') \qquad \in 
        \Hom(\HH_1,\HH_3).
    \end{equation}
    \item The partition function defined by (\ref{Z(e alpha)}) satisfies the main equation
    \begin{equation}\label{main eq on Xi flip}
        (\delta+Q)Z=0
    \end{equation}
    with $\delta$ the cellular coboundary operator on $\Xi_\mr{flip}$ and $Q$ acting on the boundary spaces of states.
\end{enumerate}
\end{thm}
\begin{proof}
    (i) is immediate by construction (\ref{Z(e alpha)}).
 (ii) follows from the $A_\infty$ relation (\ref{A_infty rel n}): $Q c_{n}$ is the sum of contractions of pairs of cyclic $A_\infty$ operations corresponding to ways to cut an $n$-gon by a single diagonal, see Figure \ref{fig:A_infty rel via cutting a polygon}. %{\color{red} [PICTURE]} 
\begin{figure}[h]
    \centering
    \includegraphics[scale=0.6]{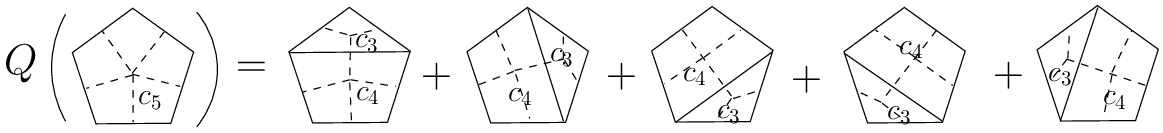}
    \vspace{0.1cm}
    \caption{$A_\infty$ relation as a sum over cuttings of a polygon. Dashed lines depict the dual ribbon graphs.}
    \label{fig:A_infty rel via cutting a polygon}
\end{figure}
 %\marginpar{Explain the argument more clearly?}
For $\alpha$ some polygonal decomposition of $\Sigma$,
we have %\marginpar{signs?}
\begin{equation}\label{proof of main eq for Xi flip}
    QZ(\alpha)= -\sum_{\mr{polygons}\;p\;\mr{of}\;\alpha} \;\;\sum_{\mr{cutting}\;p\;\mr{by\;a\;diagonal}} Z(\alpha') = -(\delta Z)(\alpha),
\end{equation}
 where $\alpha'$ is the polygonal decomposition resulting from cutting one polygon $p$ in $\alpha$ by a diagonal. In (\ref{proof of main eq for Xi flip}), the first equality is the $A_\infty$ relation while the second equality is due to the combinatorics of attachment of cells in $\Xi_\flip$.
% Cellular boundary operator on $\Xi$ is described by the same operation -- cutting a polygon by a diagonal. This leads to a cancellation in (\ref{main eq on Xi flip}).
    %by construction (\ref{Z(e alpha)}), it suffices to check (\ref{main eq on Xi flip}) for a single polygon, i.e., for $\Sigma$ a disk and $P$ collection of $n$ points on the boundary circle.
\end{proof}

\begin{example}\leavevmode 
\begin{itemize}
\item For $\alpha$ a triangulation, $Z(\alpha)$ is $Q$-closed. This is restriction of the equation (\ref{main eq on Xi flip}) to $0$-cochains. Or one can see it directly, from (\ref{A_infty rel 2}) and cyclicity of $c_2$ (or equivalently, from $Q(c_3)=0$ and $Q(g)=0$).
\item Let $\alpha$ be a polygonal decomposition where one polygon $p$ is a square and all others are triangles, and let $\alpha_1,\alpha_2$ be the two triangulations resulting from dividing $p$ by a diagonal (in two possible ways). Then $\alpha_1$ and $\alpha_2$ are related by a Pachner flip performed in the square $p$ and one has
\begin{equation}
    Z(\alpha_1)-Z(\alpha_2)=QZ(\alpha).
\end{equation}
This is the up-to-homotopy version of flip-invariance of the strict TQFT (\ref{flip invariance strict tqft}).
%If $\alpha,\alpha'$ are two triangulations related by a Pachner flip, arising from 
\item For $\cc \in C_\bt(\Xi_\flip)$ any \emph{cycle} on $\Xi_\flip$ (i.e., $\partial\cc=0$), the value of $Z$ on it is $Q$-closed:
\begin{equation}\label{QZ(c)=0}
    QZ(\cc)=0.
\end{equation}
\end{itemize}    
\end{example}

\begin{remark}
    Given two composable cobordisms $\gamma_1\xra{\Sigma'} \gamma_2$, $\gamma_2\xra{\Sigma''}\gamma_3$ each equipped with a set of points $P'$, $P''$, such that $P'\cap \gamma_2=P''\cap \gamma_2$, one has the composed cobordism $\gamma_1\xra{\Sigma}\gamma_3$ equipped with a set of points $P=P'\cup P''$. Gluing of polygonal decompositions on $\Sigma',\Sigma''$ induces a map 
    \begin{equation}\label{sewing as operation on Xi flip}
        \ccirc\colon \Xi_\flip(\Sigma'',P'')\times \Xi_\flip(\Sigma',P')\ra \Xi_\flip(\Sigma,P).
    \end{equation}
    In terms of this operation, functoriality (\ref{flip HTQFT functoriality}) can be stated as
    \begin{equation}
        Z(\zeta''\ccirc \zeta') = Z(\zeta'')\circ Z(\zeta'), 
    \end{equation}
    for $\zeta'\in C_\bt(\Xi_\flip(\Sigma',P'))$, $\zeta''\in C_\bt(\Xi_\flip(\Sigma'',P''))$ any two cellular chains. Here on the left one has the sewing operation (\ref{sewing as operation on Xi flip}) and on the right one has a composition of linear maps between spaces of states.
\end{remark}

\begin{remark}
    The construction of this section admits the following generalization. As the local data, instead of a cyclic $A_\infty$ algebra $V$ one can take a \emph{cyclic $A_\infty$ category} $\mc{V}$. Then vertices of polygonal decompositions are decorated by objects of $\mc{V}$, edges -- by morphisms and polygons -- by higher compositions. The partition function is defined by the by the same formula (\ref{Z(e alpha)}) seen as  a contraction of higher compositions in $\mc{V}$.
\end{remark}

%{\gr Remark: replace local $A_\infty$ algebra by local $A_\infty$ category (put objects at vertices, morphisms at edges, compositions on polygons).}

\begin{remark}[Degree of the metric on $V$] %\marginpar{new}
In the considerations above we were assuming that the metric $g\colon V\otimes V\ra \kk$ has degree zero.
However one can consider 
 flip theory with $V$ a cyclic $A_\infty$ algebra with metric $g$ of nonzero degree $q$. Then the partition function for a triangulated cobordism $(\Sigma,T)$ has degree $q(\#\{\mr{triangles}\}-\#\{\mr{edges}\})$. In particular, for a closed triangulated surface the partition function automatically vanishes for $q\neq 0$. Later, in the context of secondary polytopes, we have to insist that $g$ has degree zero -- otherwise, e.g., the r.h.s. of equation (\ref{Pachner 2 up to homotopy}) (compatibility with stellar subdivision/aggregation) is not a homogenenous expression.
\end{remark}

\begin{remark}
    The construction of flip theory (\ref{Z(e alpha)}) can be thought of as a modification of Kontsevich's state sum model on ribbon graphs \cite{Kontsevich_Feynman}, where (a) the construction is in terms of polygonal decompositions of surfaces rather than ribbon graphs and, more importantly, (b) surfaces with boundary are allowed.
\end{remark}

\subsubsection{The analogy between combinatorial HTQFT and continuum TCFT}\label{sss: analogy between comb HTQFT and cont TCFT}
\begin{itemize}
    \item A triangulated surface corresponds to a ``space'' of a continuum TCFT -- a surface with complex structure.
    \item A triangulated surface $\Sigma,\alpha$  with boundary (polygonal holes) corresponds to a surface with punctures. Given a choice of states $\Phi_j\in \HH(S^1_j,\alpha|_{S^1_j})$ on the boundary polygons, their pairing with the partition function
    \begin{equation}
        \langle Z(\alpha), \bigotimes_{j}\Phi_j \rangle
    \end{equation}
    corresponds to a correlator of point observables $\langle \prod_j \Phi_j(z_j) \rangle$ in continuum theory.
    
    \item We think of squares, pentagons etc. in a polygonal decomposition as ``defects'' in a triangulation, where special observables $c_4,\,c_5,\ldots$ are put. Moreover, we think of $c_4$ as $G^\mr{tot}$ -- the superpartner (or ``$Q$-primitive'') of the stress-energy tensor, cf. Section \ref{ss: TCFT} -- in continuum theory. So, e.g., evaluation of $Z$ on a surface with $n$ holes, with a polygonal decomposition which is a triangulation everywhere except $p$ square 2-cells, paired with states $\Phi_1,\ldots,\Phi_n$ placed in the polygonal holes, corresponds to the $p$-form correlator (\ref{TCFT Z}) in continuum TCFT, see Figure \ref{fig: cTCFT correlator}.  From this viewpoint, $c_3$ corresponds to the ``vacuum vector'' or ``unit field'' in TCFT (its insertion in a puncture erases the puncture). Cyclic $A_\infty$ maps $c_5,c_6,\ldots$ should be thought of as ``higher components'' of the combinatorial field $G^\mr{tot}$, not present present in a usual TCFT.\footnote{
    One can imagine $c_{\geq 5}$ appearing as
    an artifact of discretization.
    Another way these components might appear in continuum theory is if one relaxes the conservation of $G^\tot$ to conservation modulo $Q(\til G)$. In this way $\til{G}$ would correspond to $c_5$ etc.
    }
\end{itemize}
\begin{figure}[h]
    \centering
    \includegraphics[scale=0.7]{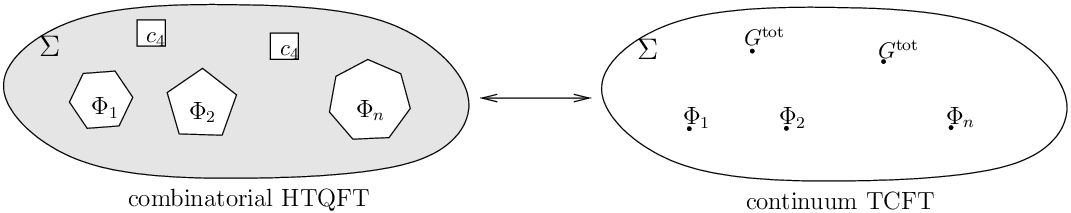}
    \caption{A correlator in combinatorial HTQFT (on the left; shaded region is the triangulated part of the surface) vs. continuuum TCFT (on the right).}
    \label{fig: cTCFT correlator}
\end{figure}

\subsection{%Simplest 
Naive 
BV operator and its problem}\label{ss: simplest BV operator}
%{\gr $\Delta$ is represented by the single square with glued opposite sides.  It is a nontrivial cycle in $\Xi_{flip}$. Its square (w.r.t. gluing of cylinders) should be homologous to zero in $\Xi_{flip}$ of the cylinder with one vertex on each boundary and one vertex in the bulk, but doesn't seem to be. Hopefully it is zero-homologous in the full $\Xi$.}

Consider the setup of Example \ref{ex: Xi flip for simple cylinder}: $\Sigma$ a cylinder with $P$ consisting of two points, one on each boundary circle. Let $\cc_\Delta$ be the 1-cycle comprised of the only 1-cell of $\Xi_\flip$ (cf. Figure \ref{fig: Xi flip for simple cylinder}). Remark that $\Xi_\flip$ is homotopy equivalent to the component $E_2^\fr(1)$ of the framed little disk operad, cf. Section \ref{sss: Efr and BV}; under this homotopy equivalence the homology class of $\cc_\Delta$ corresponds to the homology class representing the BV operator $\Delta$ that in the case of topological conformal field theory equals $G_{0,-}$. 
%So, we will refer to $\cc_\Delta$ as a representative of the BV operator
In particular, we think of the value of our HTQFT on $\cc_\Delta$ 
as a chain-level representative of the BV operator, $Z(c_\Delta) ``=" %\approx 
\Delta$:
\begin{equation}\label{simple BV operator}
    g(x,Z(\cc_\Delta)y)= (-1)^{|x|}\mr{Str}_V\, m_3(x,\bt,y)\qquad \mr{for}\; x,y\in V,
\end{equation}
with $\mr{Str}_V$ the supertrace over $V$.
Note that $Z(\cc_\Delta)$ is $Q$-closed, cf. (\ref{QZ(c)=0}).

The problem of the operator $Z(c_\Delta)$ is that it is not guaranteed to square to zero modulo $Q$-exact terms: 
\begin{equation}
    Z(\cc_\Delta)\circ Z(\cc_\Delta) \neq Q(\cdots),
\end{equation}
while the BV operator coming from the framed little disk operad does square to zero. The reason is that the cycle $\cc_\Delta\ccirc \cc_\Delta\in C_2(\Xi_\flip(\Sigma',P'))$ is \emph{not null-homologous} in $\Xi_\flip$: 
\begin{equation}
    \cc_\Delta\ccirc \cc_\Delta\neq \dd(\cdots),
\end{equation}
%for some $3$-chain $\cdots$.
see Figures \ref{fig: Xi_flip cyl bulk pt}, \ref{fig: Xi glip cyl bulk pt H_2 generator}.
Here $\ccirc$ means composition of cobordisms equipped with polygonal decompositions (\ref{sewing as operation on Xi flip}); $\Sigma'$ is again a cylinder (thought of as a composition of $\Sigma$ with itself) with $P'$ comprised of three points -- one on each boundary and one in the bulk. More precisely, homology $H_p(\Xi_\flip(\Sigma',P'))$ is concentrated in degrees $p=0,1,2$ and has rank $1$ in these degrees; $\cc_\Delta\ccirc \cc_\Delta$ is the generator of homology in degree $p=2$. 
%\marginpar{double check this statement}

\begin{figure}[h]
    \centering
    \includegraphics[scale=0.5]{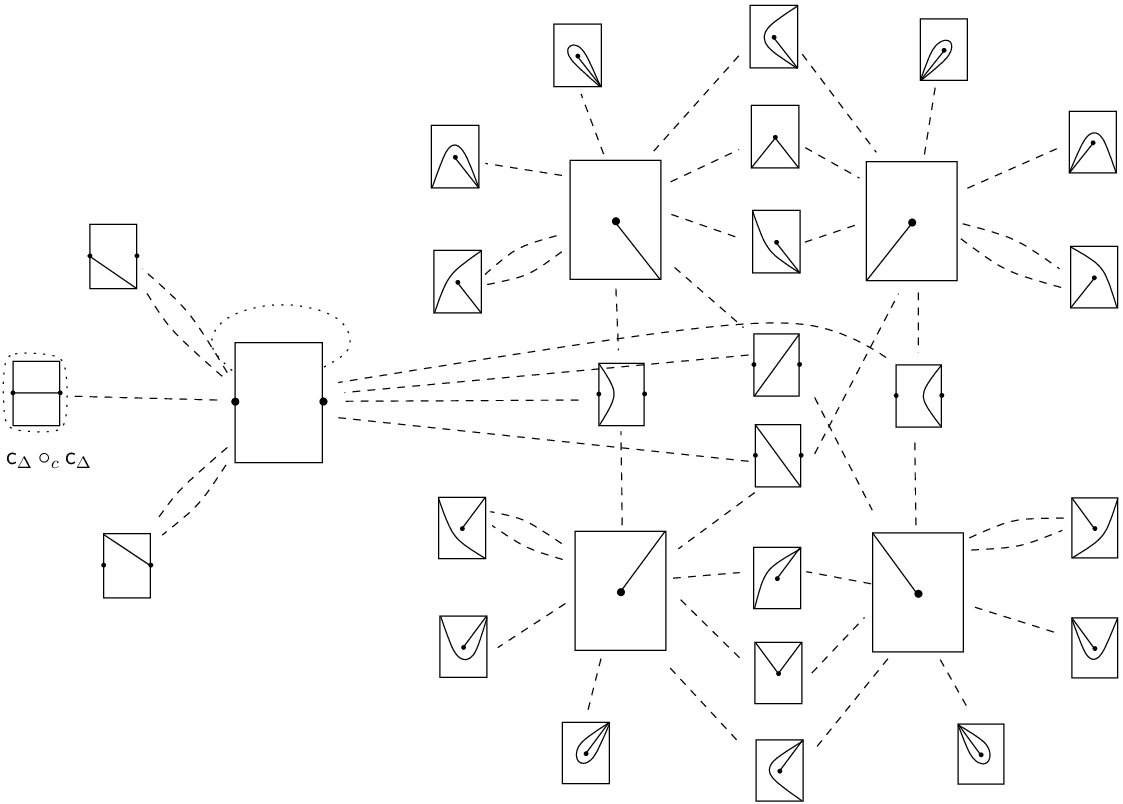}
    \caption{3-cells (larger rectangles) and 2-cells (smaller rectangles) of $\Xi_\flip$ of a cylinder with one point on each boundary and one bulk point. For each rectangle, left and right side are glued. Dashed lines show cell attachment. The leftmost 2-cell is in fact a cycle and is $\cc_\Delta\ccirc\cc_\Delta$.}
    \label{fig: Xi_flip cyl bulk pt}
\end{figure}
\begin{figure}[h]
    \centering
    \includegraphics[scale=0.7]{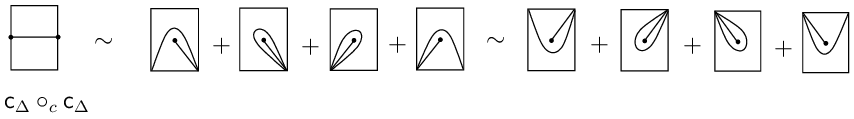}
    \caption{%Representatives of the 
    Generator of $H_2(\Xi_\flip(\Sigma',P'))$.}
    \label{fig: Xi glip cyl bulk pt H_2 generator}
\end{figure}

The core of the problem here is that taking a composition-square is an endomorphism in $E_2^\fr(1)$ (which has vanishing $H_2$) but in the setting of combinatorial HTQFT is a map $\Xi_\flip(\Sigma,P)\ra \Xi_\flip(\Sigma',P')$ where the right hand side has nontrivial $H_2$. Put another way, the problem is that changing the set $P$ for a cylinder (adding a bulk point) changes the homotopy type of $\Xi_\flip$. 
Conjecturally, this is remedied by transitioning from the flip complex $\Xi_\flip$ to the appropriate ``full'' Pachner complex $\Xi$, see Remark \ref{rem: simple BV operator squares to zero in Xi}.

%the composition in the framed little disk operad (in relevant degree) is a map $E_2^\fr(1)\times E_2^\fr(1) \ra $

\section{%Improved BV operator that squares to zero (mod $Q$) %as a differential
Combinatorial BV operators
}
\label{sec: improved BV operator}
In this section we introduce the notion of a BV cycle on $\Xi_\flip$ and the associated notion of a combinatorial BV operator that squares to zero modulo $Q$-exact terms, and show a family of examples. While the definition makes sense for any coefficient field $\kk$, in our examples we fix $\kk=\ZZ_2$ (somewhat mysteriously, for $\kk$ of characteristic $\neq 2$ we obtain signs which invalidate the main equation (\ref{c^2=dd(...)})).

\begin{definition}
Let us define a \emph{BV cycle} as a 1-cycle $\cc_\Delta$ on $\Xi_\flip$ of a cylinder with $P$ consisting of $k\geq 1$ points on each boundary circle and $m\geq 0$ bulk points, such that
%\begin{enumerate}[(i)]
%    \item $\cc_\Delta$ is not null-homologous, i.e., represents a nontrivial class in homology $H_1(\Xi_\flip)$.
%    \item 
    $\cc_\Delta\circ \cc_\Delta$ is null-homologous in $\Xi_\flip$:
\begin{equation}\label{c^2=dd(...)}
    \cc_\Delta\circ \cc_\Delta = \dd (\cdots)
\end{equation}
for some 3-chain $\cdots$.
%\end{enumerate}
We will say that a BV cycle $\cc_\Delta$ is \emph{nontrivial} if it is not null-homologous, i.e., represents a nontrivial class in homology $H_1(\Xi_\flip)$.
\end{definition}
Then the value of HTQFT on $\cc_\Delta$, 
\begin{equation}
    Z(\cc_\Delta)\in \Hom_{-1}(\underbrace{\HH(S^1_{(k)})}_{V^{\otimes k}},\underbrace{\HH(S^1_{(k)})}_{V^{\otimes k}}
    )
\end{equation}
 is $Q$-closed but (generally\footnote{For a general $A_\infty$ algebra $V$ and assuming that $\cc_\Delta$ is nontrivial. On the other hand, if $\cc_\Delta$ is null-homologous, then $Z(\cc_\Delta)$ is automatically $Q$-exact.}) not $Q$-exact and satisfies
\begin{equation}\label{Z(c_Delta)^2=Q(...)}
    Z(\cc_\Delta)^2=Q(\cdots).
\end{equation}
Here $S^1_{(k)}$ stands for a circle triangulated into $k$ intervals.
We call $Z(\cc_\Delta)$ a \emph{combinatorial BV operator}.

Note that (\ref{Z(c_Delta)^2=Q(...)}) implies that the operator that $Z(\cc_\Delta)$ induces in $Q$-cohomology squares to zero on the nose.

A non-example is the 1-cycle we considered in Section \ref{ss: simplest BV operator}: it \emph{fails} the condition (ii) above.

\subsection{The simplest BV cycle: $k=2$}
Next, consider the 1-cycle $\cc_\Delta$ in the case $k=2$, $m=0$ (two points on each boundary circle, no bulk points) shown in Figure \ref{fig: BV cycle}.
\begin{figure}[h]
    \centering
    \includegraphics[scale=0.5]{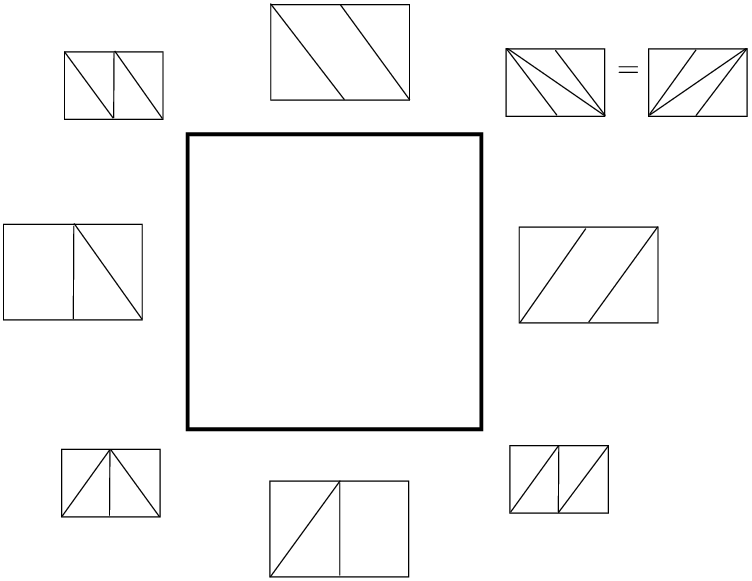}
    \caption{1-cycle in $\Xi_\flip$ of a cylinder with two points on each boundary, shown by thick lines. For each rectangle, left and right side are glued. Large rectangles are 1-cells and small ones are 0-cells.}
    \label{fig: BV cycle}
\end{figure}

\begin{proposition}\label{prop: improved BV cycle}
    The 1-cycle of Figure \ref{fig: BV cycle} is a BV cycle for coefficients $\kk=\ZZ_2$.
\end{proposition}
\begin{proof}
    %Property (i) -- nontriviality of $\cc_\Delta$ in homology -- is by direct computation. For the property (ii) 
    %One has a proof-by-picture. 
    Consider the following 3-chain $B$ in $\Xi_\flip$ shown in Figure \ref{fig: B}.
    %one can present an explicit 3-chain $B$ on $\Xi_\flip$ (homeomorphic to a solid torus, 2-disk$\times S^1$) such that $c_\Delta\circ c_\Delta=\dd B$.
    \begin{figure}[h]
        \centering
        \includegraphics[scale=0.7]{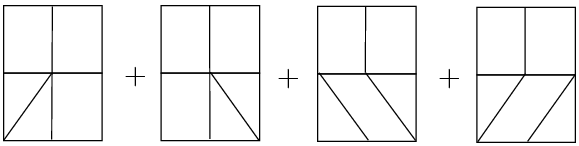}
        \caption{3-chain $B$ in $\Xi_\flip$.}
        \label{fig: B}
    \end{figure}
    By inspection, 
    %$B$ is homeomorphic to a solid torus (2-disk$\times S^1$) glued out of four 3-cubes, and 
    the boundary of $B$ in chains on $\Xi_\flip$ with $\ZZ_2$ coefficients is exactly the 2-cycle $c_\Delta\ccirc c_\Delta$.
\end{proof}

%{\gr Write $Z(\cc_\Delta)$ explicitly, as contraction of $c_4$'s and $c_3$'s.}
The combinatorial BV operator corresponding to the BV cycle of Figure \ref{fig: BV cycle} is:
\begin{multline}
    Z(\cc_\Delta)\colon t_i\otimes t_j \mapsto \Big(
    %(c_3)_{iac} (c_3)_{kab} (c_4)_{jclb}+(c_4)_{iakc} (c_3)_{lab} (c_3)_{jcb}+\\+
    %(c_3)_{kca} (c_4)_{ibla} (c_3)_{jcb}+ (c_3)_{iac} (c_4)_{jbka} (c_3)_{lbc}
    (c_3)_{kca} (c_3)_{iba} (c_4)_{jclb}+ (c_4)_{iakc} (c_3)_{jba} (c_3)_{lbc}+\\
    +(c_3)_{iac}(c_4)_{jbka} (c_3)_{lbc}+ (c_3)_{kca} (c_4)_{ibla} (c_3)_{jcb}
    \Big) t_k\otimes t_l,
\end{multline}
where $\{t_i\}$ is some orthonormal basis\footnote{
Note that an orthonormal basis in $V$ will necessarily mix degree $k$ and $-k$ components of $V$, so it cannot be homogeneous w.r.t. $\ZZ$-grading. However it can be chosen to be homogeneous w.r.t. the induced $\ZZ_2$-grading.
} in $V$ and $(c_n)_{i_1\cdots i_n}$ are the structure constants of the cyclic $A_\infty$ operations $c_n$ in this basis, see Figure \ref{fig: comb BV operator}.\footnote{
For comparison, the operator (\ref{simple BV operator}) in these notations is simply $t_i\mapsto (c_4)_{iaja} t_j$.
}
\begin{figure}[h]
    \centering
    \includegraphics[scale=0.7]{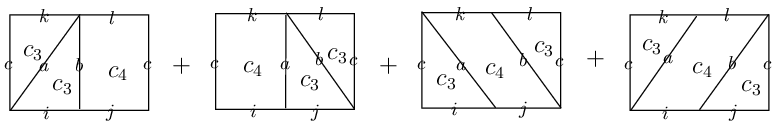}
    \caption{Combinatorial BV operator.}
    \label{fig: comb BV operator}
\end{figure}

\subsection{Generalization -- a ``long'' BV cycle}
One has a variant $\cc_\Delta^k$ of the BV cycle above with an even number $k=2p$, with $p=1,2,\ldots$, of points on each boundary circle (and no bulk points). 
\begin{comment}
\begin{figure}[H]
    \centering
    \includegraphics[scale=0.7]{long_BV_cycle.eps}
    \caption{Caption}
    \label{fig:enter-label}
\end{figure}
\end{comment}

For instance, for $k=4$ it is shown in Figure \ref{fig: c^4}. %has the following form.
\begin{figure}[h]
    \centering
    \includegraphics[scale=0.5]{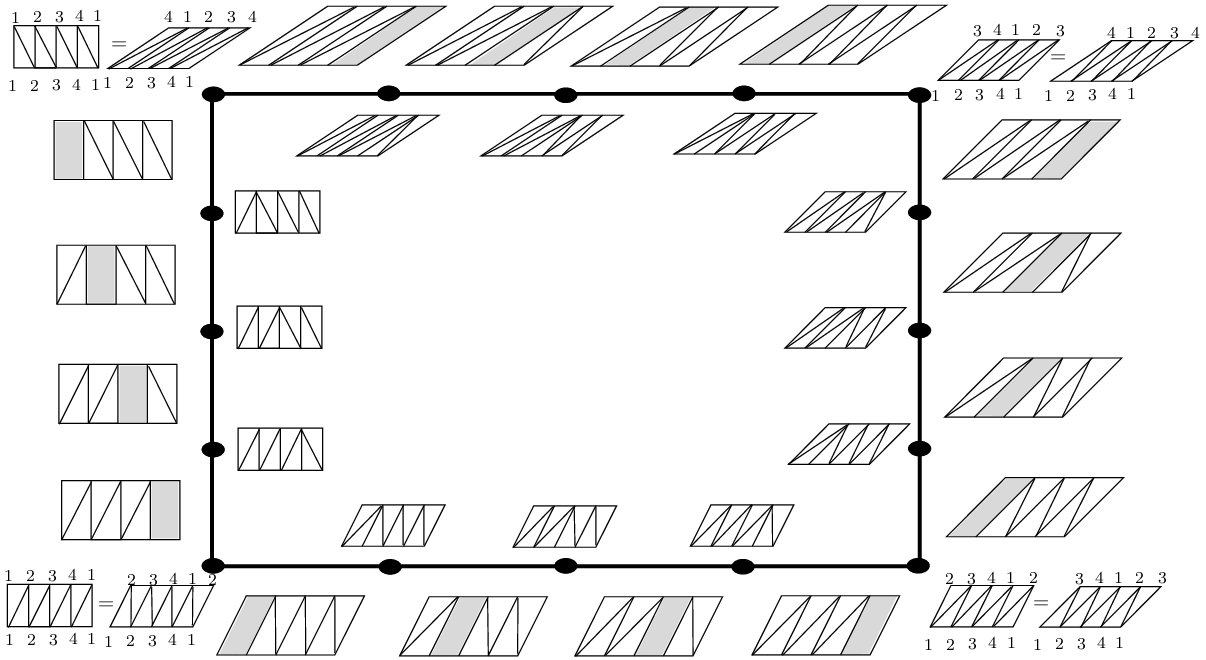}
    \caption{``Long'' BV cycle $\cc_\Delta^k$ for $k=4$. Numbers $1,2,3,4$ are the labels of boundary vertices on in/out boundary circles (vertically aligned in the picture). Square cell on the cylinder is shaded to make the picture more readable.}
    \label{fig: c^4}
\end{figure}

Generally, the cycle $\cc_\Delta^k$ consists of $k^2$ edges of $\Xi_\flip$ and has the form
\begin{equation} \label{c_Delta^n}
    \cc_\Delta^k=\sum_{i=0}^{k-1}\sum_{j=1}^k T_i\ccirc [R^{j-1} S L^{k-j}]=
    \sum_{i=0}^{k-1}\sum_{j=1}^k %\vcenter{\hbox{
    \raisebox{-0.6cm}{
    \includegraphics[scale=0.8]{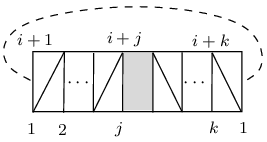}
    }
    %}}
\end{equation}
where we introduced the following notations:
\begin{itemize}
    \item $S$ is a square (shaded in the picture above), $R$ is a square subdivided into two triangles by  a SW--NE (``right'') diagonal, $L$ is a square subdivided by a SE--NW (``left'') diagonal.
    \item $[R^a S L^b]$ is a strip consisting (going left-to-right) of $a$ right-triangulated squares, a non-triangulated square, and $b$ left-triangulated squares. Leftmost and rightmost edges of the strip are identified.
    \item $T_i$ is the operation of cyclically relabeling the vertices of a triangulated circle, shifting the labels by $i \bmod k$. One can think of $T_i$ as a ``degenerate'' morphism in the category of cobordisms with polygonal decompositions, which justifies the notation $T_i\ccirc [\cdots]$, meaning the cobordism $[\cdots]$ with the labels of vertices on the out-boundary (top, in our pictures) shifted by $i\bmod k$. 
    %Below we will also need the right composition 
    Likewise (we will need it below),
    $[\cdots]\ccirc T_{-i}$ stands for shifting the labels of vertices on the  in-boundary of the cobordism $[\cdots]$ by $i\bmod k$.
\end{itemize}
To lighten notations, in the remainder of this section we will be mostly suppressing $\ccirc$ in computations.

Note that for $k=2$, the cycle (\ref{c_Delta^n}) becomes the cycle of Fig. \ref{fig: BV cycle}.

\begin{proposition}
\label{prop: c^k}
    Formula (\ref{c_Delta^n}) indeed defines a BV cycle with coefficients in $\kk=\ZZ_2$.
    %and satisfies $\cc_\Delta^k\ccirc \cc^k_\Delta=\dd(\cdots)$.
\end{proposition}
\begin{proof}
The fact that $\cc_\Delta^k$ is a cycle is checked as follows. Denote
\begin{equation}
    \mc{A}\colon=\sum_{j=1}^k [R^{j-1}SL^{k-j}].
\end{equation}
Note that we have (a) $\dd \mc{A}=[R^k]-[L^k]$ and (b) $[R^k]=T_1 [L^k]$. This implies
\begin{equation}
    \dd \cc_\Delta^k= \dd \left(\sum_{i\in \ZZ_k} T_i%\ccirc 
    \mc{A}\right)= \sum_{i\in \ZZ_k} T_i%\ccirc 
    (T_1 %\ccirc
    [L^k]-[L^k]) = \sum_{i\in \ZZ_k} T_{i+1} 
    %\ccirc
    [L^k]-T_i
    %\ccirc 
    [L^k] =0.
\end{equation}

%\begin{comment}
Intuitively, $\mc{A}$ corresponds in $E_2^\fr(1)$ to rotation of the in-circle\footnote{
Rather, in $\cc_\Delta^k$ the out- (top) circle is rotated, but it is equivalent to the opposite rotation of the in-circle.
%By a mismatch of conventions, the out- (top) circle in $\cc_\Delta^n$ corresponds to the in-circle in $E_2^\fr(1)$.
} by $2\pi/k$ (seen as a 1-chain on the space $E_2^\fr(1)$) and $\cc_\Delta^k$ corresponds to the full rotation of the in-circle, which is a 1-cycle.
%\end{comment}

The square of $\cc_\Delta^k$ is null-homologous -- it is the boundary of the 3-cycle 
\begin{equation}
B=D\ccirc 
\cc_\Delta^k
\end{equation}
with
\begin{equation}\label{D}
    D=\sum_{l=0}^{p-1}\sum_{j=2}^{k} T_{2l}%\ccirc 
    [S R^{j-2} S L^{k-j}]
    = \sum_{l=0}^{p-1}\sum_{j=2}^{k} 
    \raisebox{-0.6cm}{
    \includegraphics[scale=0.7]{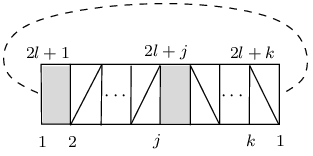}
    }
\end{equation}
Indeed, one has
\begin{multline}\label{dd D computation}
    \dd D = \sum_{l=0}^{p-1}  T_{2l} %\ccirc 
    \Big([SL^{k-1}]+[SR^{k-1}]+\sum_{j=2}^k ([R^{j-1}S L^{k-j}]+[LR^{j-2}SL^{k-j}])\Big) \\
    =\sum_{l=0}^{p-1} T_{2l} %\ccirc 
    (\sum_{j=1}^k [R^{j-1}SL^{k-j}]
    +[SR^{k-1}]+\sum_{j=2}^k [LR^{j-2}S L^{k-j}])\\
    =\sum_{l=0}^{p-1} T_{2l} %\ccirc 
    (\mc{A}+
    T_1%\ccirc 
    \mc{A} %\ccirc 
    T_{-1}) .
\end{multline}

We have
\begin{multline}
    \dd B = \dd (D%\ccirc 
    \cc^k_\Delta) = (\dd D)%\ccirc 
    \cc^k_\Delta=
    \sum_{l=0}^{p-1}\sum_{i=0}^{k-1} T_{2l}%\ccirc 
    (\mc{A}+T_1%\ccirc 
    \mc{A} %\ccirc 
    T_{-1})%\ccirc 
    T_{i}%\ccirc
    \mc{A}\\=
    \sum_{l=0}^{p-1}\sum_{i=0}^{k-1} (T_{2l}+T_{2l+1}) %\ccirc 
    \mc{A} %\ccirc 
    T_{i}%\ccirc  
    \mc{A}
    = \sum_{i',i=0}^{k-1} T_{i'}%\ccirc 
    \mc{A} %\ccirc 
    T_{i}%\ccirc 
    \mc{A}=
    \cc_\Delta^k \ccirc 
    \cc_\Delta^k
    %\frac12 (\cc_\Delta^n\ccirc\cc_\Delta^n+\cc_\Delta^n\ccirc\cc_\Delta^n)= \cc_\Delta^n\ccirc\cc_\Delta^n
\end{multline}
as claimed.
\end{proof}

\begin{remark}
Note that the construction of the 1-cycle (\ref{c_Delta^n}) makes sense for even or odd $k$ (and yields a cycle). However, in our construction of the cobounding 3-chain for $\cc_\Delta^k\ccirc \cc_\Delta^k$ we do use that $k$ is even.
\end{remark}

\begin{remark}%\marginpar{new}
    Repeating the computation for $\kk$ a field of characteristic $\neq 2$, we see that $\cc_\Delta^k$ is still a cycle, but our proof that it squares to a boundary breaks down due to signs: one obtains $\dd D= \sum_{l=0}^{p-1}T_{2l}(\mc{A}-T_1\mc{A}T_{-1})$ and hence 
    \begin{equation}
        \dd (D\cc_\Delta^k) = \sum_{i',i=0}^{k-1}\boxed{(-1)^{i'}}T_{i'}\mc{A}T_i \mc{A} \neq \cc_\Delta^k\ccirc \cc_\Delta^k.
    \end{equation}
    It is a part of Conjecture \ref{conjecture: Pachner complex} below that $(\cc_\Delta^k)^2$ is null-homologous (not just modulo 2 and not just for even $k$) in the appropriate ``full'' Pachner complex $\Xi$.
\end{remark}

\begin{conjecture}\label{conjecture: c_Delta^k is nontrivial}
    The BV cycle $\cc_\Delta^k$ defined by (\ref{c_Delta^n}) is nontrivial.
\end{conjecture}
The intuition, why $\cc_\Delta^k$ corresponds to the standard generator of $H_1(E_2^\fr(1))$, was spelled out in the proof above. A possible way to prove the Conjecture would be to show that $Z(\cc_\Delta^k)$ is nonzero for some explicit minimal cyclic $A_\infty$ algebra $V$ (``minimal'' means with $Q=0$).

\begin{remark}%\marginpar{adapt it to $\kk=\ZZ_2$?}
    For $V$ a cyclic $A_\infty$ algebra of Massey operations
    on cohomology $H^\bt(M,\ZZ_2)$ of a compact manifold $M$, 
    %by homotopy transfer from singular cochains with 
    %induced on de Rham cohomology $H^\bt(M)$ by homotopy transfer from de Rham algebra of a compact manifold $M$, 
    $Z(\cc_\Delta^k)$ is zero. Indeed, by a degree argument, in each term there will be a $c_4$ factor with $1$ as one of the inputs, hence the full expression vanishes. A similar vanishing statement holds for the algebra $H^\bt(M,\ZZ_2)$ extended by its coadjoint $A_\infty$ module.

    In particular, for a proof of the Conjecture above, one needs to look for some other examples of cyclic $A_\infty$ algebras -- e.g. those coming from Fukaya $A_\infty$ categories.
\end{remark}

%\marginpar{? write the ansatz for $\dd^{-1}(\cc^n_\Delta)^2$ for $n$ odd? }

\begin{remark}[Dressing the BV cycle by triangulated cobordisms]\leavevmode
%\begin{enumerate}[1.]
%    \item 
    Assume that $a,b$ are two triangulated cylinders with $k$ vertices on each boundary (or more generally $0$-cycles on $\Xi_\flip$) and assume additionally that the composition $b\ccirc a$ is ``rotation-invariant,'' i.e.,
    \begin{equation}\label{rotational invariance of ba}
        T_{-1} ba T_{1}=ba.
    \end{equation}
    Then
    \begin{equation}
    c\colon =a\ccirc \cc_\Delta^k \ccirc b
    \end{equation} 
    is also a BV cycle (for $\kk=\ZZ_2$). The fact that it is a cycle is obvious: $\dd c=a(\dd c_\Delta^k)b=0$. For its square, similar to the computation above, we have
    \begin{multline}\label{acb squares to zero}
        \dd (aDba \cc_\Delta^k b)= \sum_{l\in \ZZ_p} aT_{2l} \mc{A} ba \sum_{i'\in \ZZ_k} T_{i'}\mc{A} b +
        \sum_{l\in \ZZ_p} aT_{2l+1} \mc{A} \underbrace{T_{-1} ba T_1}_{ba} \sum_{i'\in \ZZ_k} T_{i'}\mc{A} b \\=
        \underbrace{a\sum_{i\in \ZZ_k} T_i \mc{A} b}_c\underbrace{a \sum_{i'\in \ZZ_k}T_{i'}\mc{A} b}_c=c\ccirc c.
    \end{multline}
    Thus, $c\ccirc c$ is null-homologous.

\begin{comment}
 %   \item  
    If we assume that $\Xi_\flip$ is connected for $\Sigma$ a cylinder and any set of vertices $P$ with $k$ vertices on each boundary (i.e. that any triangulation can be reached from any other triangulation by a sequence of flips), then we can show that $c=a \cc_\Delta^k b$ squares to a boundary, without making the assumption of rotation invariance (\ref{rotational invariance of ba}). Indeed, for any triangulated cylinder $a$ with $k$ vertices on each boundary, let
    \begin{equation}\label{averaging}
        a_\mr{sym}\colon= \frac{1}{k}\sum_{i\in \ZZ_k} T_i a T_{-i}.
    \end{equation}
    By assumption of connectedness of $\Xi_\flip$, we have that $T_i a T_{-i}$ is homologous to $a$ and thus $a_\mr{sym}$ is homologous to $a$.
    Then we have 
    \begin{multline}
        c\ccirc c=a\cc_\Delta^k b a\cc_\Delta^k b\sim  
        a_\sym\cc_\Delta^k b a\cc_\Delta^k b\sim
         a_\sym\cc_\Delta^k b_\sym a\cc_\Delta^k b \\
         \sim
    a_\sym\cc_\Delta^k b_\sym a_\sym\cc_\Delta^k b \sim
    a_\sym\cc_\Delta^k b_\sym a_\sym\cc_\Delta^k b_\sym=
    (a_\sym\cc_\Delta^k b_\sym)^2\sim
    0.
    \end{multline}
    Here $\sim$ means ``homologous'' and in the last step we used (\ref{acb squares to zero}) for the 1-cycle $a_\sym \cc_\Delta^k b_\sym$. In this argument we should assume that 
    %the coefficients are in 
    the ground field
    $\kk$ is
    has characteristic zero (or coprime with $k$), because of the denominator $k$ in (\ref{averaging}).
%\end{enumerate}
\end{comment}
\end{remark}

\subsection{A combinatorial $\mr{BV}_\infty$ operator that squares to zero exactly
(or: combinatorial replacement of the operator $Q+u G_{0,-}$)} \leavevmode \label{ss: BV_infty}
\\
\textbf{Notation:} Let $\mr{cyl}_{k\xra{m} k}$ stand for a cylinder with the set $P$ consisting of $k$ vertices on each boundary circle and $m$ bulk vertices. In particular, we think of $\mr{cyl}_{k\xra{m} k}$ as a cobordism $S^1_{(k)}\ra S^1_{(k)}$. In the case $m=0$, we will just write $\mr{cyl}_{k\ra k}$. Note that composing $s$ copies of  $\mr{cyl}_{k\ra k}$ %with itself $s$ times 
yields
\begin{equation}
    \mr{cyl}_{k\ra k}^{\ccirc s}\colon=
     \underbrace{\mr{cyl}_{k\ra k}\ccirc \cdots \ccirc  \mr{cyl}_{k\ra k}}_s
     =
    \mr{cyl}_{k\xra{k(s-1)} k}.
\end{equation}

The properties of a combinatorial BV operator 
\begin{equation}\label{c_Delta identities}
QZ(\cc_\Delta)=0, \quad Z(\cc_\Delta)^2=Q Z(B)
\end{equation}
for some 3-chain $B$ can be summarized by saying that the operator 
\begin{equation}\label{Q+ uZ(c)-u^2 Z(B)}
    Q+u Z(\cc_\Delta) - u^2 Z(B)
\end{equation}
squares to zero modulo $u^3$. Here and in the remainder of this section, $u$ is a formal parameter of degree $+2$.

This leads us to the following refinement of the notion of a combinatorial BV operator.
\begin{definition}\label{def: BV_infty}
    We define a combinatorial $\mr{BV}_\infty$ operator as a degree $+1$ element
    \begin{equation}
        \Delta_\infty \in \mr{End}(\HH(S^1_{(k)})   %,\HH(S^1_{(k)})
        ) [[u]]
    \end{equation}
    of the form
    \begin{equation}\label{Delta_infty}
        \Delta_\infty = Q+ u Z(B_1 %c_\Delta
        )+u^2 Z(B_2)+ u^3 Z(B_3)+\cdots
    \end{equation}
    with $B_s \in C_{2s-1}(\Xi_\flip(\mr{cyl}_{k\ra k}^{\ccirc s}))$, with $s=1,2,\ldots$, satisfying the equation\footnote{Note that the degree of $s$-th term in (\ref{Delta_infty}) is indeed $2s+(1-2s)=1$. }
    \begin{equation}\label{Delta_infty^2=0}
        (\Delta_\infty)^2=0.
    \end{equation}
\end{definition}
%Comparing the notations in (\ref{Delta_infty}) and (\ref{)
Given a $\Delta_\infty$ operator, $\cc_\Delta\colon=B_1$ automatically satisfies the identities (\ref{c_Delta identities}), with $B=-B_2$.

\begin{proposition}\label{prop: BV_infty}
\begin{comment}
    For any even $k=2,4,\ldots$, one has a $\mr{BV}_\infty$ operator (\ref{Delta_infty}) with 
    \begin{equation}\label{B_s}
        B_s=  (-1)^{s-1} \D^{\ccirc (s-1)}\ccirc \cc_\Delta^k,\quad s=1,2,\ldots,
    \end{equation}
    with $\cc_\Delta^k$ defined by (\ref{c_Delta^n}) and 
    \begin{equation}
        \D = \frac{1}{k} \sum_{i\in \ZZ_k} T_i D T_{-i}
    \end{equation}
    the cyclically symmetrized version of the chain (\ref{D}). 
    
    In particular, the $\BV_\infty$ operator can be written as
    \begin{equation}
        \Delta_\infty=Q+\frac{u}{1+u Z(\D)}Z(\cc_\Delta^k).
    \end{equation}
\end{comment}
For $k=2$ one has a $\mr{BV}_\infty$ operator (\ref{Delta_infty}) with coefficients in $\kk=\ZZ_2$ with 
    \begin{equation}\label{B_s}
        B_s=   D^{\ccirc (s-1)}\ccirc \cc_\Delta,\quad s=1,2,\ldots,
    \end{equation}
    with $\cc_\Delta=[SL]+[RS]+T_1[SL]+T_1[RS]$ the BV cycle of Figure \ref{fig: BV cycle} and 
    $D=[SS]$ --  the 2-cycle (\ref{D}) for $k=2$.
    
    In particular, the $\BV_\infty$ operator can be written as
    \begin{multline}
        \Delta_\infty=Q+\frac{u}{1+u Z(D)}Z(\cc_\Delta)=\\=
         Q+u Z\left( 
        \vcenter{\hbox{
\includegraphics[scale=0.5]{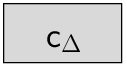}
        }}
        \right)+
        u^2 Z\left( 
        \vcenter{\hbox{
\includegraphics[scale=0.5]{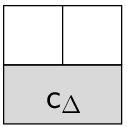}
        }}
        \right)+
        u^3 Z\left( 
        \vcenter{\hbox{
\includegraphics[scale=0.5]{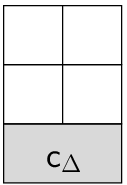}
        }}
        \right)+\cdots
    \end{multline}
    Pictures in the formula above are drawn on a cylinder: left and right sides are always identified.
\end{proposition}

\begin{proof}
    We only have to check that (\ref{Delta_infty^2=0}) holds. From (\ref{Delta_infty}) we have
    \begin{multline}\label{Delta_infty^2 computation}
        \Delta_\infty^2=
        \sum_{s\geq 1} u^s \left( QZ(B_s)+\sum_{1\leq s'< s} Z(B_{s'}) Z(B_{s-s'})\right)=\\=
        \sum_{s\geq 1} u^s Z\left(\dd B_s + \sum_{1\leq s'<s} B_{s'}B_{s-s'} \right).
        %\sum_{s_1,s_2\geq 1} u^{s_1+s_2} Z(D_\sym^{s_1-1}c_\Delta^k D_\sym^{s_2-1}c_\Delta^k).
    \end{multline}
    Note that from (\ref{dd D computation}) we have
    \begin{equation}
        \dd D=\mc{A}+T_1\mc{A}T_{-1}.
    \end{equation}
    From (\ref{B_s}), using the fact that $T_i$ commutes with $D$ and that $T_i \cc_\Delta=\cc_\Delta$, we get
    \begin{multline}
        \dd B_s = \sum_{s'=1}^{s-1}  D^{s'-1} (\dd D) D^{s-1-s'} \cc_\Delta
       =\sum_{s'=1}^{s-1}  D^{s'-1}  \underbrace{(\mc{A}+T_1\mc{A})}_{\cc_\Delta} D^{s-1-s'}\cc_\Delta\\
       =\sum_{s'=1}^{s-1}  D^{s'-1} \cc_\Delta  D^{s-1-s'} \cc_\Delta 
       =  \sum_{s'=1}^{s-1}B_{s'}B_{s-s'}.
    \end{multline}
    Thus, all terms in the r.h.s. of (\ref{Delta_infty^2 computation}) vanish.
\end{proof}

\begin{comment}
\begin{example}
    For $k=2$, one has extra simplification due to the fact that the chain $D=[SS]$ is cyclically symmetric. Thus in (\ref{B_s}) one can use $D$ instead of $\D$:
    \begin{equation}
        \Delta_\infty= 
        Q+u Z\left( 
        \vcenter{\hbox{
\includegraphics[scale=0.5]{BV_infty_term1.eps}
        }}
        \right)+
        u^2 Z\left( 
        \vcenter{\hbox{
\includegraphics[scale=0.5]{BV_infty_term2.eps}
        }}
        \right)+
        u^3 Z\left( 
        \vcenter{\hbox{
\includegraphics[scale=0.5]{BV_infty_term3.eps}
        }}
        \right)+\cdots
    \end{equation}
    with $\cc_\Delta=[SL]+[RS]+T_1[SL]+T_1[RS]$ the BV cycle of Figure \ref{fig: BV cycle}. Pictures in the formula above are drawn on a cylinder: left and right sides are always identified.
\end{example}
\end{comment}

\begin{remark}
    In continuum TCFT one has an important operator 
    \begin{equation}\label{Q equivariant}
    Q_{S^1}\colon=Q+u G_{0,-}
    \end{equation} -- the $S^1$-equivariant BRST operator acting on local fields. 
    The $\BV_\infty$ operators discussed in this section can be thought of as a combinatorial replacement of $Q_{S^1}$.  
    
    Note that the operator (\ref{Q equivariant}) squares to $u L_{0,-}$ whereas $\Delta_\infty$ squares to zero. This may be interpreted as corresponding to the combinatorial HTQFT having only fields of zero spin (thus, $L_{0,-}$ vanishes).
\end{remark}

\subsubsection{Dual language: $\mr{BV}_\infty$ chains as Maurer-Cartan elements in the dga $\mr{CYL}_k$}

%Fix an even $k=2,4,\ldots$ and 
Fix $k=1,2,\ldots$ and
consider the differential graded algebra
\begin{equation}
    \mr{CYL}_k^\bt\colon= \bigoplus_{s\geq 1} C_{-\bt} (\Xi_\flip(\cyl_{k\ra k}^{\ccirc s}))
\end{equation}
with differential given by the boundary operator $\dd$ on cellular chain on $\Xi_\flip$ and with associative multiplication $\ccirc$ induced by composition of cobordisms (stacking of cylinders). Then the Definition \ref{def: BV_infty} implies that for a $\BV_\infty$ operator $\Delta_\infty$, the constituent chains $B_s$ arrange into an element
\begin{equation}
    \mc{C}\colon= u B_1+ u^2 B_2+ \cdots \quad \in   \mr{CYL}_k[[u]]
\end{equation}
satisfying the Maurer-Cartan equation
\begin{equation}
    \dd \mc{C}+\mc{C}\ccirc \mc{C}=0.
\end{equation}
In other words, $\mc{C}$ is a Maurer-Cartan element of the dga $\mr{CYL}_k[[u]]$.

\subsection{Questions/wish list for the flip theory} %\marginpar{new}
\begin{enumerate}[(i)]
    \item Construct explicit examples of BV cycles in $\Xi_\flip$ with coefficients in $\ZZ$ rather than $\ZZ_2$. Do these examples satisfy the 7-term relation (\ref{7-term relation}) in $\Xi_\flip$ up to a boundary, with product realized by a triangulated pair of pants?
    \item Does the 7-term relation hold in $\Xi_\flip$ for the BV cycles $\cc_\Delta^k$ (with coefficients in $\ZZ_2$)?
    \item Find examples of cyclic $A_\infty$ algebras for which the combinatorial BV operator associated with the BV cycle $\cc_\Delta^k$ is nontrivial in $Q$-cohomology. In particular, this would prove Conjecture \ref{conjecture: c_Delta^k is nontrivial}.
    \item Understand the homotopy type of $\Xi_\flip$ for $\Sigma$ a surface with boundary, cf. Question \ref{question: Xi_flip vs Ribbon graphs. (b) boundary}. 
    
    In particular, for $\Sigma$ a disk, does $\Xi_\flip$ have nontrivial $H_1$ for sufficiently many bulk vertices, or does it stay contractible for any set of vertices?
    
    For $\Sigma$ a cylinder and $P$ containing no bulk points, is $\Xi_\flip$ homotopy equivalent to $E_2^\fr(1)\sim S^1$? And if so, does $[c_\Delta^k]$ map to the standard generator of $H_1$ under this equivalence?
    \item  %\marginpar{A ``physicsy'' comment}
    One may consider a class of cyclic $A_\infty$ algebras arising as an $A_\infty$ deformation of the commutative associative ring  $\mathbb{A}=O_X$ of functions on a manifold $X$ by a Maurer-Cartan element of the Hochschild dgLa of $\mathbb{A}$. One can interpret the corresponding combinatorial HTQFT as a model of topological string theory with target $X$. One can also consider $X$ to be a noncommutative space. Note that noncommutativity can in particular play the role of a regulator, making $\mathbb{A}$ finite-dimensional (e.g. taking $X$ to be the fuzzy sphere).
    \item It would be interesting to study flip theory for $V$ a \emph{minimal} cyclic $A_\infty$ algebra, i.e., with $Q=0$ (but with nontrivial $m_3$). In this case, partition functions for triangulated cobordisms are strictly (not modulo $Q$) independent of the triangulation.
\end{enumerate}

\section{Towards the full 2d HTQFT on triangulated cobordisms}
\label{sec: full 2d HTQFT on triang cob}
%\marginpar{I rewrote/expanded this intro paragraph}
In this section we explain the idea that making an appropriate replacement for $\Xi_\flip$ coming from the secondary polytope construction, may resolve the problems in the construction of combinatorial 2d HTQFT (namely, one wants compatibility with stellar subdivisions/aggregations, non-fixed set of vertices of triangulations, and a homotopy equivalence between $\Xi$ and Zwiebach's moduli space $\til\MM_{h,n}$ arising in continuum TCFT \cite{Zwiebach}).

In particular, we construct Model 2 of combinatorial 2d HTQFT -- ``secondary polytope theory.'' It is a partial model: it is only defined on a disk, realized as a convex hull of a set of points $A\subset \RR^2$. This model is invariant up to homotopy w.r.t. both Pachner moves. The appropriate local algebraic input for the model is an ``$\Ahat$ algebra'' -- a refinement of a cyclic $A_\infty$ algebra by certain extra homotopies (defined in Section \ref{ss: def Ahat algebra}).  As an aside, we present a 1d toy model of the construction yielding a generalization of the standard setting for topological quantum mechanics, which may be of independent interest.

We also discuss the ideal Model 3, conjecturally defined on arbitrary surfaces, compatible with both Pachner moves and having the expected homotopy type of $\Xi$. In particular, in this model one has a natural BV algebra structure on $Q$-cohomology.

%\subsection{First idea: relaxing the strict second Pachner move relation  (\ref{Pachner 2 rel (strict)}) up to homotopy}

\subsection{A quick recollection of secondary polytopes}

Fix $d\geq 1$ (we will be interested 
primarily %almost exclusively 
in the case $d=2$) and a finite collection of points in general linear position %\footnote{I.e. no 3 points are on the same line.} 
$A\subset \RR^d$. We assume that the affine span of $A$ is $\RR^d$. The construction of \cite{GKZ} (see \cite[Section 1]{KKS} for a detailed review) assigns to $A$ the ``secondary polytope'' -- a convex polytope  $\Sp(A)\subset \RR^A/\mr{Aff}(\RR^d)\simeq \RR^{|A|-d-1}$ (with $\mr{Aff}(\RR^d)$ the space of affine linear functions on $\RR^d$)  such that: 
\begin{itemize}
    \item Vertices of $\Sp(A)$ correspond to \emph{regular} triangulations of the convex hull $\mr{Conv}(A)\subset \RR^d$ with vertices in $A$ 
    %(not all points in $A$ have to be used as vertices of the triangulation).
    (there is no requirement that all points of $A$ have to be used as vertices of the triangulation). 
    A triangulation is ``regular'' if there exists a convex continuous %piecewise linear 
    function on $\Conv(A)$, linear on each simplex and  ``breaking'' (non-differentiable) on the codimension one faces of the triangulation. More precisely, the vertex $e_T$ corresponding to a triangulation $T$ is 
    \begin{equation}\label{secondary polytope e_T}
   e_T= \pi(a\mapsto \mr{vol}(\mr{Star}(a)))
    \end{equation}
    -- the function assigning to a point $a\in A$ the volume of its star in $T$, seen as an element of $\RR^A$ and then projected by the quotient map $\pi\colon \RR^A\mapsto \RR^A/\mr{Aff}(\RR^d)$. Thus, the secondary polytope is the convex hull of points (\ref{secondary polytope e_T}):
    \begin{equation}
        \Sp(A)=\Conv(\{e_T\}).
    \end{equation}
    \item Faces of $\Sp(A)$ correspond to polyhedral subdivisions %decompositions 
    of $\Conv(A)$. More precisely, a face $e_\alpha$ of $\Sp(A)$ corresponds to a collection $\alpha$ of subsets $A_i$ of $A$ such that
\begin{equation}
\Conv(A)=\cup_i \Conv(A_i)
\end{equation} 
is a \emph{regular} polyhedral subdivision of $\Conv(A)$, i.e., a presentation as a union of convex $d$-dimensional polyhedra (with non-empty interior) overlapping over faces. Here ``regular'' again means that there exists a convex continuous %piecewise linear 
function on $\Conv(A)$, linear on the constituent polyhedra and breaking along faces. The face $e_\alpha$ is combinatorially equivalent\footnote{
I.e., has the same poset of faces.
} to the product of secondary polytopes
\begin{equation}
    e_\alpha\simeq \prod_i \Sp(A_i).
\end{equation}
In particular for the dimension of the face $e_\alpha$ one has 
\begin{equation}
    \dim e_\alpha=\sum_i (|A_i|-d-1).
\end{equation}
We will refer to points of $A_i$ in the interior of $\Conv(A_i)$ as ``floating points.''
\end{itemize}

%\marginpar{\color{green} Construction via sections of $p\colon \Delta^{|A|-1}\ra \RR^d$?}

\begin{remark}
In the terminology of \cite{KKS}, for $A\subset \RR^d$ a set of points in general position and  $\p=\Conv(A)$ a convex polytope,\footnote{When we say ``convex polytope,'' we understand that it should have non-empty interior.} the pair $(\p,A)$ is called a ``marked polytope.'' For $A'\subset A$, $\p'=\Conv(A')$ a convex polytope, the pair $(\p',A')$ is called a ``marked subpolytope'' of $(\p,A)$. A polyhedral subdivision of $(\p,A)$ is a collection of marked subpolytopes $\{(\p_i,A_i)\}$ such that $\p=\cup_i \p_i$ is a polyhedral subdivision, i.e., $\p_i\cap \p_j$ is either empty or a common face of $\p_i$ and $\p_j$.
\end{remark}

One has a poset $R(\p,A)$ of polyhedral subdivisions of $(\p,A)$, with order given by refinement. The minimal elements of $R(\p,A)$ are regular triangulations. The unique maximal element is the tautological polyhedral subdivision, with $\{(\p_i,A_i)\}$ consisting of just $(\p,A)$ itself. The face poset of the secondary polytope $\Sp(A)$ coincides with $R(\p,A)$.

\begin{example}
    For $d=2$, $A\subset \RR^2$ the set of vertices of a convex polygon, $\Sp(A)$ is combinatorially equivalent\footnote{
    In the following we will usually omit ``combinatorially equivalent'' when describing a model for the secondary polytope.
    } to Stasheff's associahedron $K_{|A|-1}$. 
    
   % For instance, for $|A|=4$ the picture is that of Figure \ref{fig: 1-cell of Xi_flip}, ignoring the shaded region. For $|A|=5$, see the right part of Figure \ref{fig: 2-cells in Xi_flip}, again ignoring the shaded region.
   For instance, for $|A|=4$ and $|A|=5$ the pictures of secondary polytopes are those of Figure \ref{fig: 1-cell of Xi_flip} and the right part of Figure \ref{fig: 2-cells in Xi_flip}, ignoring the shaded regions.
\end{example}

\begin{example}
    %For $d=2$ and $A=\{x_1,x_2,x_3,x_4\}$ a quadruple of points such that $x_4$ is in the interior of the triangle $\sigma=\Conv(\{x_1,x_2,x_3\})$, the secondary polytope is an interval, with one endpoint decorated by tautological triangulation of $\sigma$, the other -- with stellar subdivision of $\sigma$ with vertex at $x_4$. The interval itself is decorated with tautological triangulation of $\Sigma$ with $x_4$ a floating point.
    For any $d$ and $A=\{x_0,\ldots,x_d,y\}$ a $(d+2)$-tuple of points such that 
    $y$ is in the interior of the $d$-simplex $\sigma=\Conv(\{x_0,\ldots,x_d\})$,
    the secondary polytope $\Sp(A)$ is an interval, with one endpoint decorated by tautological triangulation of $\sigma$, the other -- with stellar subdivision of $\sigma$ with vertex at $y$. The interval itself is decorated with tautological triangulation of $\sigma$ with $y$ a floating point, see Figures \ref{fig: sp stellar 1d}, \ref{fig: secondary polytope - triangle with one floating point}. 
    \begin{figure}[h]
        \centering
        \includegraphics[scale=0.7]{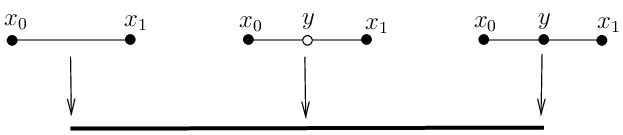}  
        %\qquad
         %       \includegraphics[scale=0.5]{sp_stellar_2d.eps}
        \caption{Secondary polytope for three points $A=\{x_0,x_1,y\}$ in $\RR$ with $x_0<y<x_1$. %Floating point is drawn as a hollow vertex. 
        For the left endpoint of $\Sp$, $\{A_i\}=\{\{x_0,x_1\}\}$. 
        For the bulk, $\{A_i\}=\{\{x_0,x_1,y\}\}$ with $y$ the floating point, shown as a hollow vertex.  For the right endpoint, $\{A_i\}=\{\{x_0,y\},\{x_1,y\}\}$.}
        \label{fig: sp stellar 1d}
    \end{figure}
    \begin{figure}[h]
        \centering
                \includegraphics[scale=0.7]{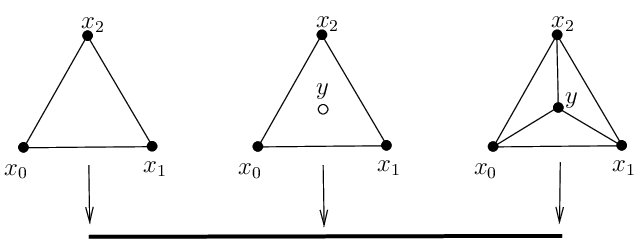}
        \caption{Secondary polytope for four points in $\RR^2$ in non-convex position.}
       \label{fig: secondary polytope - triangle with one floating point}
    \end{figure}
\end{example}

\begin{remark} %\marginpar{Is it true for $d>2$?}
    %For $d=1,2$, 
    For any $d\geq 1$,
    the edges of a secondary polytope correspond to Pachner moves between regular triangulations. In particular, for $d=2$, both Pachner flips and stellar subdivisions/aggregations can 
    appear as edges of $\Sp(A)$.
\end{remark}

\begin{example}\label{ex: secondary polytope for d=1}
    For $d=1$ and any finite set of distinct points $A\subset \RR$, the secondary polytope $\Sp(A)$ is %combinatorially equivalent to 
    a cube of dimension $|A|-2$, see Figures \ref{fig: sp 1d 2 floating pts}, \ref{fig: sp 1d 3 floating pts}. More explicitly, assuming $A$ is the ordered set $x_0<x_1<\cdots<x_n<x_{n+1}$ in $\RR$, faces of $\Sp(A)$ correspond to making for each $j=1,\ldots,n$ a choice whether 
    \begin{enumerate}[(i)]
        \item $x_j$ appears in $\cup_i A_i$ as a vertex of the triangulation $\Conv(A)=\cup_i \Conv(A_i)$, or
        \item $x_j$ appears as a floating point -- interior point of $\Conv(A_i)$ for some $i$, or
        \item $x_j$ is not in $\cup_i A_i$.
     \end{enumerate}
     The dimension of the face is the number of floating points among $x_1,\ldots,x_n$.
    \begin{figure}[h]
        \centering
        \includegraphics[scale=0.5]{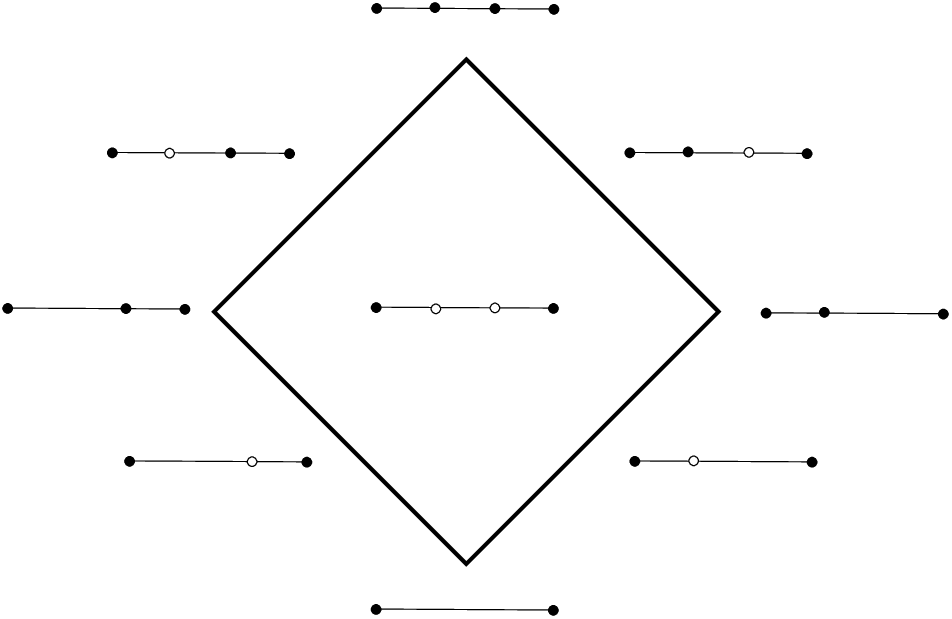}
        \caption{Secondary polytope for $A$ consisting of four points in $\RR$.}
        \label{fig: sp 1d 2 floating pts}
    \end{figure}
    %\marginpar{correct picture: transparent floating points}
    \begin{figure}[h]
        \centering
        \includegraphics[scale=0.5]{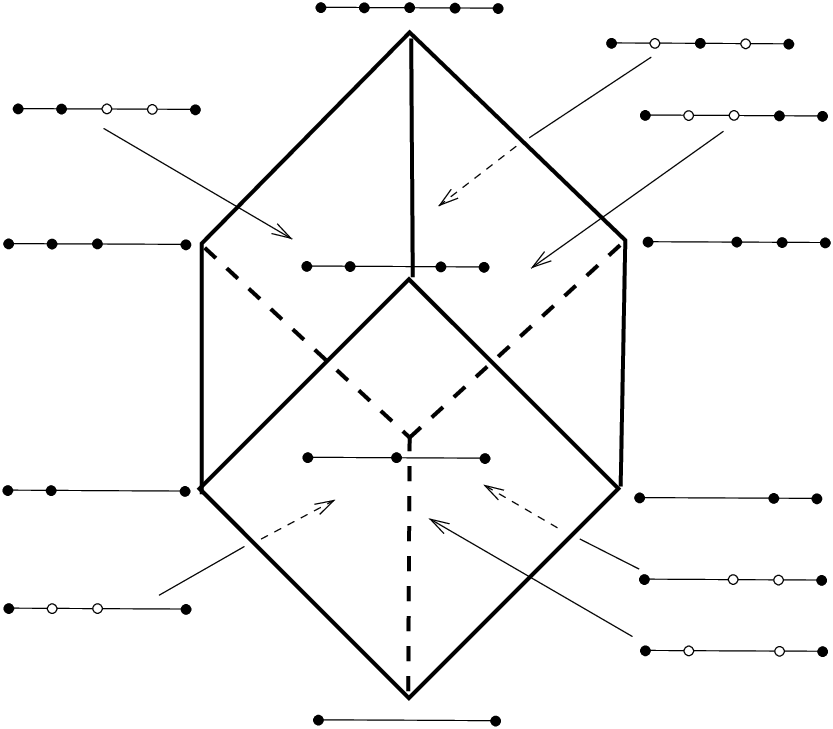}
        \caption{Secondary polytope for $A$ consisting of five points in $\RR$. Decorations of vertices and 2-faces are shown.}
        \label{fig: sp 1d 3 floating pts}
    \end{figure}
\end{example}

\begin{figure}[h]
    \centering
    \includegraphics[scale=0.5]{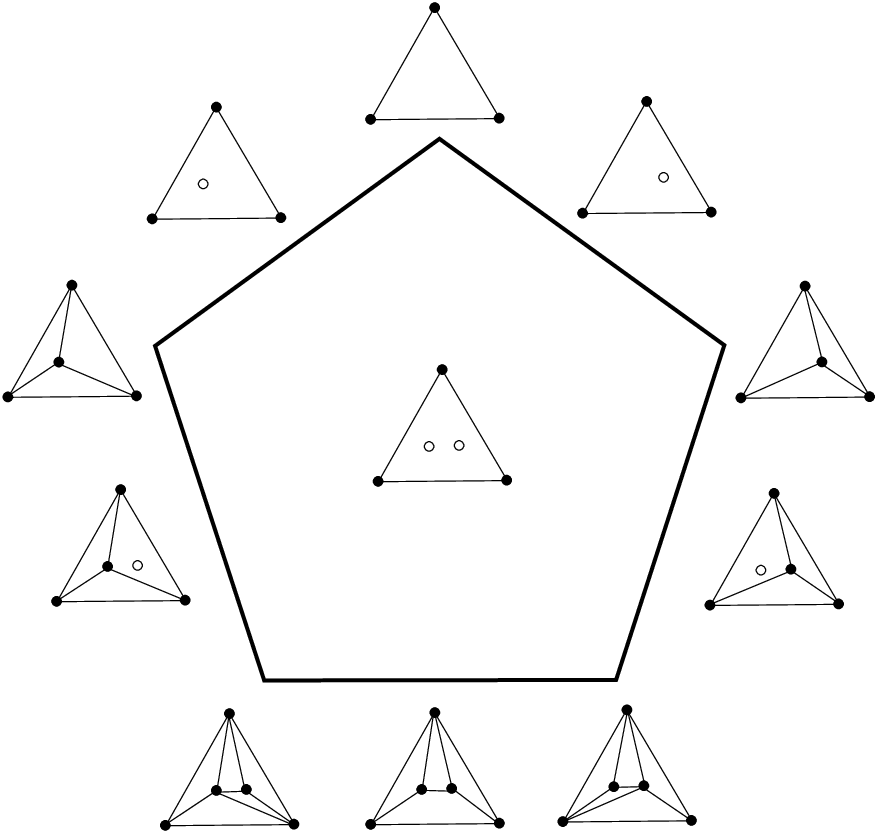}
    \caption{Example: secondary polytope for $A$ consisting of five points in $\RR^2$, with $\Conv(A)$ a triangle. The secondary polytope depends on convex geometry of the configuration; the situation shown corresponds to the two floating points in a certain open set (``chamber'') in the configuration space of two points.}
    \label{fig: SP triangle with 2 floating points}
\end{figure}

\begin{figure}[h]
    \centering
    \includegraphics[scale=0.5]{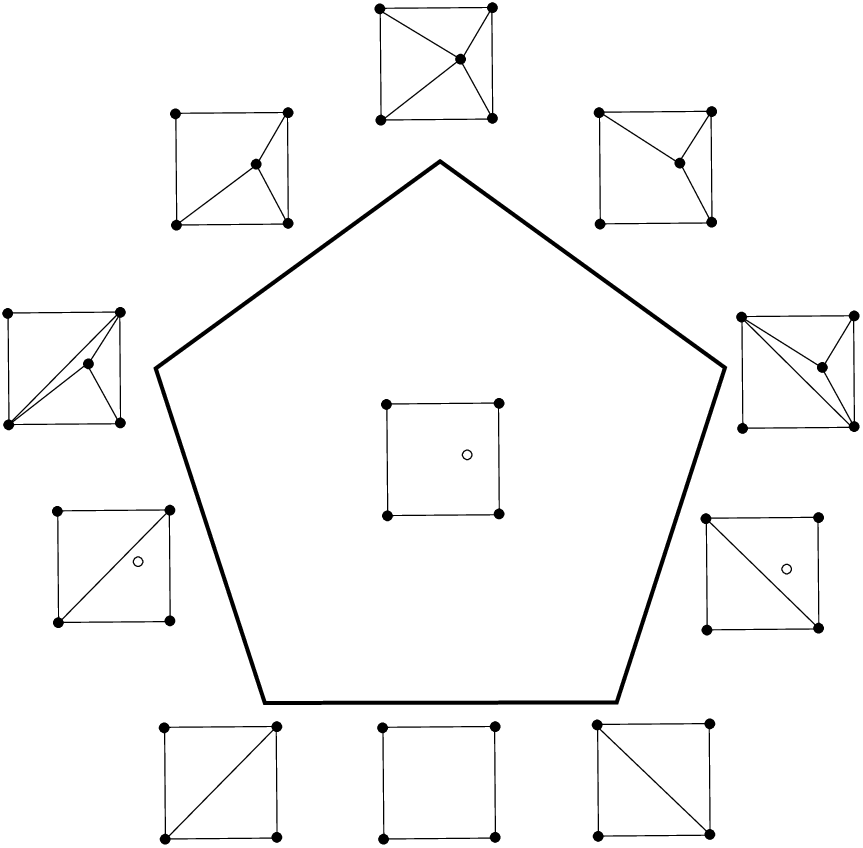}
    \caption{Example: secondary polytope for a square with one floating point in $\RR^2$. Again, combinatorics of the face poset depends on the ``chamber'' in which the floating point is located. Chambers in this case are the connected components of the square interior with the two diagonals removed.}
    \label{fig: SP square with 1 floating point}
\end{figure}

\subsubsection{Admissible subsets of $\RR^d$}
Let us denote $\Adm_m(\RR^d)$ the set of ``admissible configurations of $m$ points'' -- subsets  $A\subset \RR^d$ with $|A|=m$ such that:
\begin{enumerate}[(i)]
    \item $A$ is in general linear position and the affine span of $A$ is $\RR^d$.
    \item $A$ is not in the locus of ``exceptional wall-crossing,'' see \cite[Section 5.2]{KSS}.\footnote{Configurations satisfying condition (i) form an open dense subset of $(\RR^d)^m/S_m$. However, as one travels inside a connected component of this subset, the combinatorial type of the secondary polytope $\Sp(A)$ might change (``jump''). Locus where such change happens are the ``exceptional walls.''}
\end{enumerate}
The set $\Adm_m(\RR^d)$ is an open dense subset of $(\RR^d)^m/S_m$.
%\marginpar{Quotient by $S_m$? I.e., are these configurations of an unordered set of point?}
We will call its connected components ``configuration chambers.''
For $d=2$, $\Adm_m(\RR^2)\subset (\RR^2)^m/S_m$ is known (see \cite{KSS}) to be $\RR$-Zarisski open.

For an admissible configuration $A\subset \RR^d$, we denote $\langle A\rangle$ the number of codimension one faces of the polytope $\Conv(A)$. We will also denote  $\Adm_{m,n}(\RR^d)$ the set of admissible configurations of $m$ points in $\RR^d$ such that $\langle A\rangle=n$.

For instance for $d=2$, elements of $\Adm_{m,n}(\RR^2)$ correspond (via transitioning to the marked polygon  $(\p=\Conv(A),A)$) to convex $n$-gons in $\RR^2$ with $m-n$ ``floating points'' inside.

\subsection{%The secondary polytope operad $\wh{A}_\infty$
The definition of an $\Ahat$ algebra
}
\label{ss: def Ahat algebra} \leavevmode \\
%In this section we specialize to the dimension $d=2$. 
\begin{comment}
One has a cyclic dg operad $\wh{A}_\infty$ with generators associated to ``configuration chambers'' -- elements of $\pi_0 \Adm_{m,n}(\RR^2)$. A configuration $A\in \Adm_{m,n}(\RR^2)$ gives rise to a cyclic operation $c_{[A]}$ with $n$ cyclically ordered inputs and no outputs, of degree $3-m$ (here $[A]$ stands for the configuration chamber of $A$).  
%class of $A$ in $\pi_0 \Adm_{m,n}(\RR^2)$).
The differential $Q$ in $\Ahat$ acts as
\begin{equation}
    Q c_{[A]}= c_{[\dd\, \Sp(A)]}.
\end{equation}
\end{comment}

\textbf{Operations.}
We will call an $\Ahat$ algebra (or ``$\SP_2$ algebra'') a cochain complex $V$ over $\kk$ with differential $Q$ equipped with a compatible metric $g$ and a family of multilinear maps (``operations'')
%cyclic \marginpar{Edit/clean up the use of ``cyclic''}
%operations\footnote{``Cyclic'' means (i) with output in the ground field and (ii) the $\ZZ_n$ equivariance property spelled out below.}
enumerated by ``configuration chambers'' -- elements of $\pi_0 \Adm_{m,n}(\RR^2)$. 
A configuration $A\in \Adm_{m,n}(\RR^2)$ is assigned %gives rise to 
a  multilinear map
\begin{equation}\label{Ahat operation}
    c_{[A]}\colon V^{\otimes n}\ra \kk
\end{equation}
of degree $3-m$
(here $[A]$ stands for the configuration chamber of $A$). We understand that inputs in $V$ are assigned to the edges of the polygon $\Conv(A)$.

\begin{comment}
To be more precise, the operation (\ref{Ahat operation}) depends on $[A]$ plus a choice of one edge of $\Conv(A)$ as a ``starting edge.'' Then we assign inputs of the operation starting from that edge and going counterclockwise. %Cyclicity property is that 
%The assignment (\ref{Ahat operation}) should be equivariant under $\ZZ_n$, where $\ZZ_n$ acts by changing the choice of the starting edge and by cyclically permuting the factors of $V^{\otimes n}$.

Note that if the configuration chamber $[A]$ is invariant under rotations of $\RR^2$ by angles $2\pi k/n$ (e.g., if $A$ is $n$ points 
%in vertices of a regular $n$-gon, or more generally $n$
in convex position\footnote{Then $A$ can be moved to the set of vertices of a regular $n$-gon, staying within the same confuration chamber.}), then the operation $c_{[A]}$ is automatically invariant under cyclic permutations of its $n$ inputs.
\end{comment}

\textbf{(Partial) invariance under cyclic permutations of inputs.}
We will say that a configuration chamber $[A]$, for $A\in \Adm_{m,n}(\RR^2)$, is ``$\ZZ_q$-symmetric,'' for $q>1$ a divisor of $n$, if one can (i) move $A$ staying in the same configuration chamber so that $\Conv(A)$ is a regular polygon and then (ii) move the configuration of interior points to the same configuration rotated by the angle $2\pi/q$ while staying in the same configuration chamber.  If $[A]$ is $\ZZ_q$-symmetric, we impose the condition that the operation $c_{[A]}$ is invariant under the subgroup $\ZZ_q\subset \ZZ_n$ acting by cyclic permutations of inputs.
%.... {\color{red} [FINISH]}

\begin{example}
    If $A$ is an $n$-tuple of points in $\RR^2$ in convex position, then $[A]$ has $\ZZ_n$ symmetry. If $A$ is as in Figure \ref{fig: secondary polytope - triangle with one floating point}, $[A]$ has $\ZZ_3$ symmetry. If $A$ is as in Figure \ref{fig: SP triangle with 2 floating points} or Figure \ref{fig: SP square with 1 floating point}, then $[A]$ has no $\ZZ_q$ symmetry. The chamber of the configuration of two floating points in a square in Figure \ref{fig: conf chamber with Z2 symmetry} has $\ZZ_2$ symmetry.
    \begin{figure}
        \centering
        \includegraphics[scale=0.7]{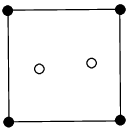}
        \caption{Configuration chamber in $\pi_0\mr{Adm}_{6,4}(\RR^2)$ with $\ZZ_2$ symmetry.}
        \label{fig: conf chamber with Z2 symmetry}
    \end{figure}
\end{example}

\textbf{Differential.}
The differential $Q$ %in $\Ahat$ 
acts on operations as
\begin{equation} \label{Ahat Q relation}
    Q c_{[A]}= -c_{[\dd\, \Sp(A)]}.
\end{equation}
Here on in the right we take the boundary of the secondary polytope of $A$. The codimension one faces comprising it correspond to special (``coarse,'' in the terminology of \cite{KKS}) polygonal subdivisions $\{(\p_i=\Conv(A_i),A_i)\}$ of $(\p=\Conv(A),A)$. For each such subdivision we compute the contraction of tensors $c_{[A_i]}$ assigned to subpolygons, and the inverse metric $g^{-1}$ assigned to the edges, and then sum over the boundary faces of $\Sp(A)$. I.e., in more detail, (\ref{Ahat Q relation}) is:
\begin{equation} \label{Ahat Q relation 2}
    Q c_{[A]}= -\sum_{\alpha=\{(\p_i,A_i)\} \;\mr{coarse \;subdivision\;of\;} (\p,A)}
    \left\langle \bigotimes_i c_{[A_i]}, \bigotimes_{\mr{edges\;of\;}\alpha} g^{-1} \right\rangle.
\end{equation}

\begin{remark}
    Specializing the structure to $n$-tuples of points in $\RR^2$ in convex position (i.e. vertices of a convex $n$-gon), for $n\geq 3$, we recover the definition of a cyclic $A_\infty$ algebra, with (\ref{Ahat Q relation 2}) the $A_\infty$ relation (cf. Figure \ref{fig:A_infty rel via cutting a polygon}). We will denote the operation corresponding to a convex $n$-gon with no floating points by $c_n$, to be compatible with notations in Section \ref{ss: local data: cyclic A_infty algebra}.
    
    Thus, an $\Ahat$ algebra is in particular a cyclic $A_\infty$ algebra and should be thought of as a refinement of the latter by certain extra homotopies. 
\end{remark}

\begin{example}
    Let $A$ be a quadruple of points in $\RR^2$ in non-convex position (a triangle with a floating point inside). 
    %cf. Figure \ref{fig: secondary polytope - triangle with one floating point}. 
    Let us denote the corresponding degree $-1$ operation $\mu\colon=c_{[A]}\colon V^{\otimes 3}\ra \kk$.
    The relation (\ref{Ahat operation}) for it takes form
    \begin{equation}\label{Pachner 2 up to homotopy}
        (Q\mu)(x,y,z) = \mr{Str}_V m_2(x,m_2(y,m_2(z,\bt)))-c_3(x,y,z),\qquad x,y,z\in V,
    \end{equation}
    where the two terms in the r.h.s. correspond to the two endpoints of the secondary polytope in Figure \ref{fig: secondary polytope - triangle with one floating point}; $m_2$ stands for the operation $c_3$ with one input converted into output by $g^{-1}$. 
    
    Note that relation (\ref{Pachner 2 up to homotopy}) is an up-to-homotopy version of the relation (\ref{Pachner 2 rel (strict)}) corresponding to second Pachner move (stellar subdivision/aggregation). The corresponding homotopy, witnessing the failure of the invariance under the move to hold strictly, is $\mu$ -- the $\Ahat$ operation for a triangle with a floating point.
\end{example}

\begin{remark}
The first few operations in an $\Ahat$ algebra, arranged by degree and arity, are as follows.

\noindent
\begin{tabular}{c|c|c|c}
 degree $\downarrow$; arity $\ra$    & $3\ra 0$ & $4\ra 0$ & $5\ra 0$  \\ \hline
   0  & $c_3=c\left(\left[\vcenter{\hbox{\includegraphics[scale=0.35]{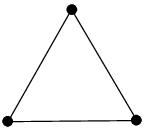}}}\right]\right)$
   & 
   & \\
   -1 & $\mu=c\left(\left[\vcenter{\hbox{\includegraphics[scale=0.35]{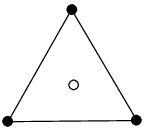}}}\right]\right)$ & 
   $c_4=c\left(\left[\vcenter{\hbox{\includegraphics[scale=0.35]{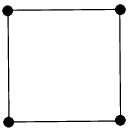}}}\right]\right)$ 
   & \\
   -2 &
   $c\left(\left[\vcenter{\hbox{\includegraphics[scale=0.35]{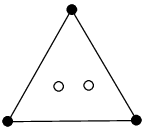}}}\right]\right)$
   & $c\left(\left[\vcenter{\hbox{\includegraphics[scale=0.35]{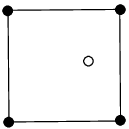}}}\right]\right)$ &
   $c_5=c\left(\left[\vcenter{\hbox{\includegraphics[scale=0.35]{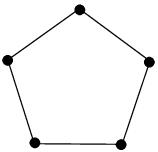}}}\right]\right)$
\end{tabular}

Note that the left and middle operations in the last row here do not have any $\ZZ_q$ symmetry, whereas the remaining operations listed in the table are invariant w.r.t. cyclic permutations of inputs.
%$\mu$ is $\ZZ_3$-symmetric and each $c_n$ is $\ZZ_n$-symmetric.

%Note that the operation on the left in the last row here is not invariant w.r.t. cyclic permutations of inputs, 
%\marginpar{EDIT}
%since a rotation of $\RR^2$ by angle $2\pi/3$ moves the configuration to a different configuration chamber. Similar point applies to the middle operation in the last row. The remaining operations listed in the table are invariant w.r.t. cyclic permutations of inputs.
\end{remark}

\begin{remark}
    An $\Ahat$ algebra can be formalized as an algebra over a ``cyclic'' dg operad $\Ahat$ freely generated by the generators $c_{[A]}$ assigned to elements of $\pi_0 \Adm_{m,n}(\RR^2)$, with the differential defined by (\ref{Ahat Q relation}), (\ref{Ahat Q relation 2}). Here ``cyclic'' means that operations have output in the ground field
    %, are invariant w.r.t. cyclic permutations of inputs 
    and that there is an extra structure (contraction with $g^{-1}$) allowing to convert any input into an output. 

    Since the operad $\Ahat$ is based on the combinatorics of secondary polytopes for subsets of $\RR^2$, we will also call this operad $\SP_2$.
\end{remark}

\subsection{ Combinatorial HTQFT on a convex polygon based on secondary polytopes (``Model 2'')
% Combinatorial HTQFT on a convex polygon
}
%\marginpar{On a polygon or on a disk -- which is better?}
\label{ss: Model 2}

Let $V$ be an $\Ahat$ algebra and $A\in \Adm_{m,n}(\RR^2)$ an admissible configuration of points on the plane. Then one has a combinatorial HTQFT on the polygon $\Sigma=\Conv(A)$ with $A$ the set of admissible vertices of polygonal decompositions of $\Sigma$. This HTQFT is the cochain on the secondary polytope
\begin{equation}
Z\in C^\bt(\Sp(A))\otimes \Hom(V^{\otimes n},\kk)
\end{equation}
defined by
\begin{equation}\label{Z model 2}
    Z(e_\alpha)= \Big\langle \bigotimes_i c_{[A_i]} , \bigotimes_{\mr{edges\;of\;}\alpha} g^{-1} \Big\rangle
\end{equation}
for $\alpha=\{(\p_i,A_i)\}$ a polygonal subdivision of $(\Sigma,A)$. 
\begin{thm}\label{thm: model 2}
The cochain $Z$ defined by (\ref{Z model 2}) satisfies the main equation of HTQFT,
\begin{equation}\label{main eq model 2}
    (\delta+Q)Z=0.
\end{equation}
\end{thm}
\begin{proof}
    It follows from the relation (\ref{Ahat Q relation 2}) in the algebra $V$, similarly to Theorem \ref{thm: model 1}.
\end{proof}

This construction realizes the ``full'' 2d combinatorial HTQFT on a convex polygon, with Pachner flips between triangulations and stellar aggregations allowed, with a fixed set of admissible vertices.\footnote{In this context one can think of the set of admissible vertices $A$ as ``aether'' -- a background geometric structure.} 
%(I.e., one can remove vertices via a stellar aggregation, but one cannot add )

An encouraging property of this construction is that $\Sp(A)$ is contractible (in particular, it is homotopy equivalent to $E_2^\fr(0)$), whereas $\Xi_\flip$ for a disk with $N$ bulk vertices is expected to have nontrivial homology depending on $N$, cf. Section \ref{sss: polygonal decompositions, ribbon graphs and the moduli space}.

\subsection{A toy model: combinatorial 1d HTQM based on secondary polytopes}
\label{ss: 1d HTQM from secondary polytopes}
\subsubsection{The definition of an $\SP_1$ algebra}
Similarly to %the definition of 
Section \ref{ss: def Ahat algebra}, we can define an algebraic structure corresponding to secondary polytopes for subsets of $\RR^1$.

We define a $\SP_1$ algebra as a cochain complex $V$ with differential $Q$ and a collection of endomorphisms $U_n\colon V\ra V$ of degree $-n$ (the operation corresponding to a subset $A$ of $n+2$ points of $\RR$, with the corresponding marked polytope being an interval $\Conv(A)\subset \RR$ with $n$ floating points), for $n=0,1,2,\ldots$, subject to the relation\footnote{
Understanding that $Q$ acts on endomorphisms of $V$ via adjoint action, we use notations $QU_n$ and $[Q,U_n]$ interchangeably (in the first notation, $Q$ acts on the operation, in the second -- on the output and on the input), cf. footnote \ref{footnote 16}.
}
\begin{equation}\label{SP_1 relation}
    QU_n= \sum_{i=0}^{n-1} (-1)^i (U_i U_{n-i-1} -U_{n-1} ).
\end{equation}
This formula is an explicit form of (\ref{Ahat Q relation 2}) for $A\subset \RR^1$; the $2n$ terms in the r.h.s. correspond to the codimension one faces of the $n$-cube $\Sp(A)$, cf. Example \ref{ex: secondary polytope for d=1}.
In particular, the first relations are:
\begin{align}
    QU &= 0,\\
    QU_1 &= U^2-U, \label{QU_1=U^2-U} \\
    Q U_2 &= U U_1-U_1 U, \\
    Q U_3 &= U U_2 - U_1^2 + U_2 U - U_2,\\
    \nonumber &\cdots
\end{align}
where we denoted $U\colon=U_0$.

One can introduce the corresponding notion of $\SP_1$ operad -- a dg operad freely generated by unary operations $U_n$, with differential acting as in (\ref{SP_1 relation}).

Note that equation (\ref{QU_1=U^2-U}) is the up-to-homotopy version of the equation (\ref{Z^2=Z}) in strict topological quantum mechanics.

\subsubsection{Combinatorial HTQM}
Fix a set $A\subset \RR$ and an $\SP_1$ algebra $(V,Q,\{U_n\}_{n\geq 0})$. We define a combinatorial HTQM as a cochain $Z$ on the secondary polytope $\Sp(A)$ with values in $\mr{End}(V)$, as follows.  
%assigning to a face $\alpha=\{(\p_i,A_i)\}$ of $\Sp(A)$ -- 
For $\alpha=\{(I_i,A_i)\}$ a triangulation of the marked interval $(I=\Conv(A),A)$ into marked sub-intervals (we assume that the intervals $I_i$ are ordered from left to right in $\RR$), the value of $Z$ on the corresponding face $e_\alpha$ of the secondary polytope is 
\begin{equation}\label{comb HTQM Z}
    Z(e_\alpha)=\overleftarrow{\prod_i} U_{|A_i|-2}\qquad \in \mr{End}(V).
\end{equation}
By virtue of the relation (\ref{SP_1 relation}), we have the main equation of HTQFT:
\begin{equation}
    (\delta+Q)Z=0.
\end{equation}

\begin{example}
    For instance, the value of HTQM on an interval $I$ subdivided into $N$ sub-intervals with no floating points is $U^N$. We think of these sub-intervals as ``elementary'' (or ``atomic'') intervals.
%    \marginpar{OK? Edit?}
    $U_1$ is the chain homotopy between HTQM on an interval consisting of (triangulated into) two elementary ones and a single elementary interval, cf. (\ref{QU_1=U^2-U}).

    Another example of formula (\ref{comb HTQM Z}), for an interval triangulated into four sub-intervals, one of them with a floating point:
    \begin{equation}
        Z\left(
        \vcenter{\hbox{\includegraphics[scale=0.7]{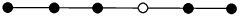}}}
        \right) = U U_1 U^2.
    \end{equation}
\end{example}

\begin{example}
    Assume we are given a continuum HTQM as in Example \ref{Ex: HTQM}, with the space of states being a cochain complex $V,Q$ and with the Hamiltonian $H=[Q,G]$, with $G\in \mr{End}(V)$ a degree $-1$ endomorphism satisfying $G^2=0$.

    Fix $T>0$ -- the length of ``elementary/atomic interval.'' Then we can construct an $\SP_1$ algebra structure on the complex $(V,Q)$ as follows:
    \begin{align}
        U&=e^{TH}, \\
        U_1&= \int_T^{2T} e^{tH+dt\, G}=
        \int_T^{2T} dt\, G e^{tH} 
        = G\frac{e^{2TH}-e^{TH}}{H},\\
        U_{\geq 2} & =0.
    \end{align}
\end{example}

\begin{example}[``Infinitesimal'' $\SP_1$ algebra]\label{ex: infinitesimal SP_1 algebra}
Fix a cochain complex $(V,Q)$.
    Consider the following ansatz for an $\SP_1$ algebra
    \begin{multline}
        U=1+TH+ O(T^2),\; U_1= T G+ O(T^2),\; U_2=T^2 G_2+O(T^3),\\ U_3=T^2 G_3+O(T^3),\;
        U_4 = T^3 G_4+O(T^4),\ldots, U_n = T^{[n/2]+1}G_n + O(T^{[n/2]+2}),\ldots
    \end{multline}
    with $H=G_0$, $G=G_1$, $G_2$, etc. a sequence of endomorphisms of $V$, with $G_n$ of degree $-n$. Here $T$ is a formal infinitesimal parameter. Then the $\SP_1$ structure equation (\ref{SP_1 relation}) considered in the leading order in $T$, yields the sequence of equations
    \begin{align}
        [Q,H]&= 0,\\
        [Q,G] &= H, \\
        [Q,G_2] &= [H,G], \\
        [Q,G_3] &= G_2-G^2 ,\\
        [Q,G_4] &= [H,G_3]-[G,G_2],\\
        [Q,G_5] &= G_4-[G,G_3],\\
        [Q,G_6] &= [H,G_5]-[G,G_4]+[G_2,G_3],\\
        [Q,G_7] &= G_6-[G,G_5]-G_3^2,
        \\ \nonumber
        &\vdots 
        %\\
        %[Q,G_n] &= \sum_{i=0}^{n-1} (-1)^i (G_i G_{n-i-1}-G_{n-1}), \label{SP_1 infinitesimal relation}
        %\\ \nonumber
        %&\vdots
    \end{align}
    Generally, the relations are:
    \begin{align}
         [Q,G_{2k}] &= \sum_{i=0}^{k} (-1)^i [G_i, G_{2k-i-1}], \label{SP_1 infinitesimal relation 1} \\
         [Q,G_{4k+1}] &= G_{4k}-\sum_{i=0}^{k-1}[G_{2i+1}, G_{4k-2i-1}], \label{SP_1 infinitesimal relation 2}\\
         [Q,G_{4k+3}] &= G_{4k+2}-\sum_{i=0}^{k-1} [G_{2i+1},G_{4k-2i+1}]-\frac12 [G_{2k+1},G_{2k+1}], \label{SP_1 infinitesimal relation 3}
    \end{align}
    for $k\geq 0$.
    We call the collection of endomorphisms $H=G_0,G=G_1,G_2,\ldots$ subject to relations above an ``infinitesimal $\SP_1$ algebra.'' 

    Note that in the special case $G_2=G_3=\cdots=0$, we recover the data of a continuum HTQM of Example \ref{Ex: HTQM}.

    Note also that the relations above involve only supercommutators. 
    %(note that for $n$ even, $G_n^2=\frac12 [G_n,G_n]$). 
    Thus, it is tempting to consider the dg Lie algebra $\mathfrak{sp}_1$ with generators $Q,\{G_n\}_{n\geq 0}$ and differential $[Q,-]$, subject to relations  (\ref{SP_1 infinitesimal relation 1})--(\ref{SP_1 infinitesimal relation 3}). In this terminology, an infinitesimal $\SP_1$ algebra is a representation of the  dgLa $\mathfrak{sp}_1$ on a complex $(V,Q)$.
\end{example}

\subsection{
Tentative Model 3 on triangulated surfaces. Conjecture on existence of the proper Pachner CW complex}
In this section we sketch a conjectural combinatorial model of 2d HTQFT -- ``Model 3'' or ``Ideal %Dream 
HTQFT'' 
-- combining the features of flip theory of Sections \ref{sec: flip theory}--\ref{sec: improved BV operator} (exists on surfaces, contains BV operators)  and of the HTQFT on a convex polygon from secondary polytopes of Section \ref{ss: Model 2} (compatible with both Pachner moves, the number of vertices is not fixed, the homotopy type of $\Xi$ is under control).

\subsubsection{Pachner complex}
\begin{conjecture}\label{conjecture: Pachner complex}
Let $\Sigma$ be a surface of genus $h$ with $n$ boundary circles. Assume that the boundary circles are triangulated into $k_1,\ldots,k_n$ intervals. We conjecture that there exists a CW complex $\Xi(\Sigma,\{k_i\})$ -- ``Pachner complex'' -- such that:
\begin{enumerate}[(i)]
    \item Vertices of $\Xi$ correspond to ``admissible'' triangulations of $\Sigma$, inducing the given triangulation on the boundary.\footnote{We are blackboxing the appropriate notion of admissibility. In the setting of secondary polytopes, it is regularity -- existence of a convex piecewise linear function inducing the given triangulation.}
    \item Edges of $\Xi$ correspond to Pachner moves between triangulations -- flips or stellar subdivisions/aggregations.
    %\item The assignment $\Sigma\mapsto \Xi(\Sigma)$ is functorial.
    \item \label{conj: Pachner complex (iii)}
    $\Xi$ is homotopy equivalent to Zwiebach's moduli space $\til\MM_{h,n}$ (cf. Section \ref{ss: TCFT}). In particular, for genus $h=0$, $\Xi$ is homotopy equivalent to the component $E_2^\fr(n-1)$ of the framed little disk operad. 
    
    We will denote the induced isomorphism on homology 
    \begin{equation}\label{f_*}
    f_*\colon H_\bt(\Xi)\ra H_\bt(E_2^\fr(n-1)).
    \end{equation}
    We assume that this isomorphism is compatible with gluing of cobordisms/operadic composition in $E_2^\fr$.
    \item Faces of $\Xi$ correspond to ``admissible'' polygonal decompositions of $\Sigma$, inducing the given triangulation of the boundary, into polygons with $\geq 3$ sides equipped with extra geometric data -- a linear chart and a configuration chamber of $\geq 0$ \emph{floating points} (plus the boundary vertices of the polygon). 
    
    Sub-faces of a face correspond to subdivisions of a given polygonal decomposition, locally described by the secondary polytope combinatorics.
    \item For $\Sigma$ a cylinder with $k_1=k_2=\colon k$, % even, 
    the BV cycle $\cc_\Delta^k$ (\ref{c_Delta^n}) is a representative of the generator of $H_1(\Xi,\ZZ)\stackrel{f_*}{\simeq} H_1(E_2^\fr(1),\ZZ)=\ZZ$.
\end{enumerate}
\end{conjecture}

In the rest of this section we will assume the Conjecture above to hold.

\begin{remark}
    One can consider a modification of the setup above where the triangulations of boundary circles are not fixed and there are extra 1-cells of $\Xi$ corresponding to  ``elementary shelling'' moves at the boundary.
\end{remark}

\begin{remark}
The Conjecture might require the assumption that $k_i$ are large enough. 

On a related note, it is possible that the appropriate ``admissibility'' property of polygonal decompositions includes in particular that they are regular as CW complexes, which would imply $k_i\geq 2$. It might be necessary that triangulations are simplicial complexes in a strict sense, which would imply $k_i\geq 3$.
\end{remark}

\begin{remark}\label{rem: simple BV operator squares to zero in Xi}
    If $k_i=1$ is allowed, we expect the cycle $\cc_\Delta$ of Section \ref{ss: simplest BV operator} to satisfy (v) above. In particular, we expect it to square to a boundary in $\Xi$ even though it does not square to a boundary in $\Xi_\flip$.

    A similar point applies to $\cc_\Delta^k$ with any odd $k=3,5,\ldots$: we do not know if it squares to a boundary in $\Xi_\flip$ but (v) above implies that it does in $\Xi$.
\end{remark}

\subsubsection{HTQFT on the Pachner complex}
Fix an $\Ahat$ algebra $(V,Q,\{c_{[A]}\})$. Then one can define a $\Hom(\HH_\ii,\HH_\oo)$-valued cochain $Z$ (the HTQFT) on $\Xi$ by the formula (\ref{Z model 2}):
\begin{equation}\label{Z model 3}
   Z(e_\alpha)= \Big\langle \bigotimes_{\mr{polygons\,}(\p_i,A_i)\,\mr{of}\,\alpha} c_{[A_i]} , \bigotimes_{\mr{edges}\backslash\{\mr{in-edges}\}} g^{-1} \Big\rangle,
\end{equation}
for $\alpha$ a polygonal decomposition of the surface and $e_\alpha$ the corresponding face of $\Xi$.
\begin{corollary}[of Conjecture \ref{conjecture: Pachner complex}]\label{corollary: Model 3}
    \begin{enumerate}[(i)]
        \item Formula (\ref{Z model 3}) defines a functor from the category of cobordisms $\Sigma$ equipped with admissible polygonal decompositions (decorated with local configuration chambers) to the category of vector spaces and linear maps.
        \item The cochain (\ref{Z model 3}) satisfies the main equation of HTQFT,
        \begin{equation}
            (\delta+Q)Z=0.
        \end{equation}
    \end{enumerate}
\end{corollary}
\begin{proof}
Part (i) is similar to (i) of Theorem \ref{thm: model 1} and is immediate from the local construction of the cochain (\ref{Z model 3}). Part (ii) is
due to the combinatorics of the face poset of $\Xi$. 
%the cochain $Z$ automaticaly satisfies the main equation of HTQFT (\ref{main eq model 2}).
\end{proof}

\subsubsection{Chain-level BV algebra structure} %\leavevmode \\
Fix  $k$ %an even $k=2,4,\ldots$ 
and denote $\HH=\HH(S^1_{(k)})=V^{\otimes k}$.
For all triangulated cobordisms in this section we will assume that the induced triangulation of each boundary circle is into $k$ intervals.

\textbf{Product (or ``closed string product").} For $\Sigma$ a pair of pants equipped with any triangulation $T$, the isomorphism $f_*$ (\ref{f_*}) identifies the homology class of $T$ in $H_0(\Xi)$ with the ``product'' generator of $H_0(E_2^\fr(2))$, cf. Section \ref{sss: Efr and BV}. Evaluating $Z$ on $(\Sigma,T)$, we obtain a degree $0$ map
\begin{equation}\label{closed string product}
    *_T=Z(T)\colon %\HH(S^1_{(k_1)})\otimes \HH(S^1_{(k_2)}) \ra \HH(S^1_{(k_3)})
    \HH\otimes \HH \ra \HH
\end{equation}
-- the ``closed string product'' (as opposed to the $A_\infty$ operation $m_2\colon V\otimes V\ra V$ -- the ``open string product'').

%In the case $k_1=k_2=k_3$, 
The operation (\ref{closed string product}) is associative modulo $Q$-exact terms as follows from the fact that two triangulations $T\circ_1 T$, $T\circ_2 T$ of a sphere with four holes can be transformed one into the other by a sequence of Pachner moves. A similar argument (comparing $T$ to its rotation by $\pi$, $T'$, which swaps the two in-circles) shows commutativity of (\ref{closed string product}) modulo $Q(\cdots)$.

\textbf{BV operator.} Evaluating $Z$ on the chain (\ref{c_Delta^n}) yields a degree $-1$ BV operator 
\begin{equation}\label{Delta model 3}
\Delta=Z(\cc^k_\Delta)\colon %\HH(S^1_{(k)})\ra \HH(S^1_{(k)})
\HH\ra \HH
\end{equation} 
which satisfies $\Delta^2=Q(\cdots)$.

\textbf{BV bracket.} Consider a triangulated pair of pants $(\Sigma,T)$ 
%with $k_1=k_2=k_3=k$ 
and consider a 1-cycle on $\Xi$ given by composing $T$ with the cycle $\cc_\Delta^k$ in all possible ways:
\begin{equation}
    R=\cc_\Delta^k\circ T- T\circ_1 \cc_\Delta^k - T\circ_2 \cc_\Delta^k.
\end{equation}
The isomorphism $f_*$ maps the homology class of $R$ to the ``BV bracket'' generator of $H_1(E_2^\fr(2))$. Evaluating $Z$ on this 1-cycle, we get a degree $-1$ map
\begin{equation}\label{BV bracket model 3}
    \{,\}_{T,\cc^k_\Delta}=Z(R)\colon %\HH(S^1_{(k)})\otimes \HH(S^1_{(k)}) \ra \HH(S^1_{(k)})
    \HH\otimes \HH\ra \HH
\end{equation}
-- the (chain-level) BV bracket.
Due to relations in homology of $E_2^\fr$, $Z(R)$ satisfies Jacobi and bi-derivation property modulo $Q$-exact terms. The latter is equivalent to the BV operator $\Delta$ satisfying the 7-term relation (\ref{7-term relation}) modulo $Q$-exact terms.

\textbf{Unit.} For $(D,S)$ a triangulated disk, the class of $S$ in $H_0(\Xi)$ maps by $f_*$ to the unit class in $H_0(E_2^\fr(0))$. Thus we have a (chain-level) ``unit''
\begin{equation}\label{Unit model 3}
   1_{S}= Z(S)\in %\HH(S^1_{(k)}).
    \HH.
\end{equation}
%Due to relations in homology of $E_2^\fr$, the objects (\ref{closed string product}), (\ref{Delta model 3}), (\ref{BV bracket model 3}), (\ref{Unit model 3}) satisfy the relations of $E_2^\fr$
The unit properties -- (\ref{Unit model 3}) being a unit for the product (\ref{closed string product}) and being annihilated by $\Delta$ -- hold up to $Q(\cdots)$, by the same reason as above.

\begin{figure}[h]
    \centering
    \includegraphics[scale=0.7]{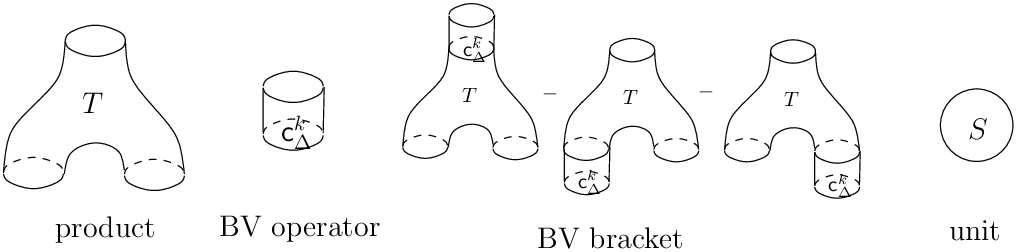}
    \caption{Chain-level BV algebra operations. $T$ is a triangulation of a pair of pants, $S$ is a triangulation of a disk.}
%    \label{fig:enter-label}
\end{figure}

\subsubsection{Reduction by passing to $Q$-cohomology and by attaching triangulated cylinders. The strict BV algebra.}
All the objects (\ref{closed string product}), (\ref{Delta model 3}), (\ref{BV bracket model 3}), (\ref{Unit model 3}) are $Q$-closed and change by $Q$-exact terms when triangulations $T,S$ are changed (or the BV cycle is changed). In particular, passing to cohomology of $Q$, one has a strict BV algebra structure on the $Q$-cohomology of the space of states, $H_Q(\HH)$.

%\begin{remark}
\textbf{Reduction by triangulated cylinders.}
    Evaluating $Z$ on any triangulated cylinder $(\Sigma,T)$ and passing to $Q$-cohomology, we obtain a map $P\colon H_Q(\HH)\ra H_Q(\HH)$. Since the composition square of $(\Sigma,T)$ can be moved to $(\Sigma,T)$ by a series of Pachner moves, $P$ is a projector: $P^2=P$. By construction, the strict BV algebra on $H_Q(\HH)$ factors through $P$, as one can attach a triangulated cylinder to any end of any of the triangulated cobordisms giving the BV algebra operations. 
    
    Put another way, one has a splitting
    \begin{equation}
        H_Q(\HH)= \underbrace{\HH^\mr{red}}_{\mr{im}(P)}\oplus \mr{ker}(P).
    \end{equation}
    The operations of the strict BV algebra descend to $\HH^\rd$ and vanish on $\mr{ker}(P)$.
    Note that $\HH^\rd$ here plays similar role to the Hochschild homology appearing in Section \ref{sss: strict TQFT: space of states for non-triangulated circle}.\footnote{Also, our construction of the space $\HH^\rd$ is similar to the construction of the space of states as the quotient by the kernel of the partition function for a cylinder in \cite{TV}.} 
%\end{remark}

Furthermore, we note that the space $\HH^\rd$ and the BV algebra structure on it are independent of $k$. The isomorphism $\HH^\rd_k\ra \HH^\rd_{k'}$ (here we indicate explicitly the triangulation of the circle as a subscript) is given by any triangulated cylinder with the in- and out-circles triangulated into $k$ and $k'$ intervals. 

\begin{remark}
Restricting to $Q$-cohomology and restricting to triangulated cobordisms only (i.e. with no higher polygons and no floating points) yields the %usual (strict) 2d Dijkgraaf-Witten model 
strict TQFT
of Section \ref{sec: combinatorial strict 2d TQFT}, based on the noncommutative graded Frobenius algebra $H_Q(V)$, with the strictly associative multiplication inherited from the operation $m_2$ on $V$. %(which is only associative up to homotopy).

In particular, as in Section \ref{sss: strict TQFT: space of states for non-triangulated circle}, one has a model for the reduced space of states $\HH^\rd$ as the center of $H_Q(V)$,
\begin{equation}
    \HH^\rd=\mc{Z}(H_Q(V)).
\end{equation}

%    The graded space $H_Q(V)$ has a structure of strictly associative noncommutative Frobenius algebra, with the product inherited from the operation $m_2$ 
\end{remark}

To summarize the discussion above, we have the following.
\begin{corollary}[of Conjecture \ref{conjecture: Pachner complex}]\label{cor: BV algebra on the reduced space of states}
    For an $\Ahat$ algebra $(V,Q,\{c_{[A]}\})$, there is a canonically defined BV algebra structure induced on the center $\mc{Z}(H_Q(V))$ of the cohomology of $Q$.
\end{corollary}

%Evaluating $Z$ on any triangulated pair of pants, we obtain a chain-level version of the 

\begin{comment}
{\gr We have a version where we can change the number of points on each boundary circle (Pachner shellings).

We expect $\Xi$ in genus zero to have the homotopy type of $E_2^\fr$ operad. This would imply the relations generators, e.g., 7-term relation.
In higher genus, we expect $\Xi$ to have the homotopy type of Zwiebach's $\til\MM_{h,n}$ -- a torus bundle over the non-compactified moduli space.
}
\end{comment}

%\subsection{The correct local data: algebra over the secondary polytope operad $\wh{A}_\infty$}

%\subsection{The partition function as a cocycle on the Pachner complex}

%\section{Open questions: 
%homotopically commutative multiplication and 
%the 7-term relation}

%{\gr 
%[REMARK: punctures $\neq$ vertices. We conjecture that adding vertices does not change the homotopy type]}

%\section*{Data availability statement} 
%Data sharing not applicable to this article as no datasets were generated or analyzed during the current study.

%\section*{Conflict of interest statement}
%On behalf of all authors, the corresponding author states that there is no conflict of interest.

%\section{Working materials}

%\includepdf[pages=-]{sketch.pdf}

%\includegraphics[scale=0.1]{BV.jpg}

%\includepdf[pages=-]{note_9_23_2023.pdf}

%\includepdf[pages=-]{note_4_24_2023.pdf}

\end{document}